\documentclass[12pt,a4paper]{article}
\usepackage{amsfonts}
\usepackage{amssymb}
\usepackage{graphicx}
\usepackage{caption}
\usepackage{subcaption}
\usepackage[fleqn]{amsmath}
\usepackage[margin=0.9in]{geometry}
\usepackage{float}
\usepackage[authoryear]{natbib}
\usepackage[title]{appendix}
\usepackage[dvipsnames]{xcolor}
\usepackage{hyperref}
\usepackage{listings}

\newtheorem{theorem}{Theorem}

\newtheorem{assumption}{Assumption}

\newtheorem{definition}{Definition}

\newtheorem{lemma}{Lemma}

\newenvironment{proof}[1][Proof]{\noindent \textbf{#1.} }{\ \hfill $\square$ \bigskip}
\definecolor{darkgreen}{rgb}{0.0, 0.2, 0.13}
\definecolor{webgreen}{RGB}{19,138,40}
\definecolor{webbrown}{RGB}{30,30,30}
\hypersetup{
colorlinks=true, linktocpage=true, pdfstartpage=1, pdfstartview=FitV,
breaklinks=true, pdfpagemode=UseNone, pageanchor=true, pdfpagemode=UseOutlines,
plainpages=false, bookmarksnumbered, bookmarksopen=true, bookmarksopenlevel=1,
hypertexnames=true, pdfhighlight=/O, urlcolor=webbrown, linkcolor=RoyalBlue, citecolor=webgreen}
\input{tcilatex}
\begin{document}

\title{A Nearly Similar Powerful Test for Mediation\newline
}
\author{Kees Jan van Garderen \\
Amsterdam School of Economics\\
University of Amsterdam\\
K.J.vanGarderen@uva.nl \and Noud van Giersbergen \\
Amsterdam School of Economics\\
University of Amsterdam\\
N.P.A.vanGiersbergen@uva.nl}
\date{}
\maketitle

\begin{abstract}
This paper derives a new powerful test for mediation that is easy to use.
Testing for mediation is empirically very important in psychology,
sociology, medicine, economics and business, generating over 100,000
citations to a single key paper. The no-mediation hypothesis $H_{0}:\theta
_{1}\theta _{2}=0$ also poses a theoretically interesting statistical
problem since it defines a manifold that is non-regular in the origin where
rejection probabilities of standard tests are extremely low. We prove that a
similar test for mediation only exists if the size is the reciprocal of an
integer. It is unique, but has objectionable properties. We propose a new
test that is nearly similar with power close to the envelope without these
abject properties and is easy to use in practice. Construction uses the
general varying $g$-method that we propose. We illustrate the results in an
educational setting with gender role beliefs and in a trade union sentiment
application.\bigskip

\end{abstract}


\centerline{Version October, 2021}

\newpage

\section{Introduction}

This paper derives a new powerful test for mediation that is easy to use.
Testing for mediation effects is empirically extremely important in various
scientific disciplines. A key paper in psychology, \citet{Baron1986} has
more than 100,000 citations\footnote{%
\ Cited by 106,782 on 13 October 2021, 90,147 on 15 January 2020, and 79,205
on 22 October 2018.} and is used in many other fields. Mediation testing is
important in accounting, e.g. \citet{Coletti2005}, marketing, e.g. %
\citet{Mackenzie1986}, sociology, e.g. \citet{Alwin1975} who used the
expression indirect effect, a term commonly used in economics also.For a
recent overview of mediation in economics see \citet{huber2020mediation} and
e.g. \cite{Heckman2015a, Heckman2015b} on treatment effects and production
technology, and \citet{imbens2020potential} who extensively reviews
connections between directed acyclic graphs (DAGs), potential outcomes,
causal inference, instrumental variables, and mediation. %
\citet{frewen2013perceived} is exemplary for the increasing network
literature, including DAGs, using mediation. The minimal selection here is
hardly representative for the vast body of literature on mediation analysis.
It only illustrates the breadth of its empirical relevance. Tests for
mediation can have extremely low power, especially when the effect is small,
or estimated with large variance. The primary purpose of this paper is to
provide a new and more powerful test that is easy to use.

The aim of mediation testing is to discover if an independent variable ($X$)
causes a dependent variable ($Y$) via an intervening, or mediating variable (%
$M$). The mediating variable is exogenous in the common experimental setting
in psychology and other fields, but is also considered exogenous in other
settings where assignments are random or constitute a natural experiment.
The basic model is simply: 
\begin{align}
Y& =\tau X+\theta _{2}M+u,  \label{eq:unrstricted_Y} \\
M& =\theta _{1}X+v,  \label{eq:unrstricted_I}
\end{align}%
where all variables are taken in deviation from their means, or more
generally after partialing out other exogenous effects. The disturbances $u$
and $v$ are assumed to be independent because of an experimental set up and
more generally because no influence of $Y$ on $M$ is assumed in this type of
model. This independence is a crucial identification condition, since the
parameter $\theta _{2}$ cannot be estimated consistently if $M$ is
endogenous. We make a convenient distributional assumption: $\left(
u_{i},v_{i}\right) ^{\prime }\sim II\emph{N}\left( 0,diag\left( \sigma
_{11},\sigma _{22}\right) \right) $, $i=1,\cdots ,n,$ with $n$ the number of
observations. This facilitates a likelihood analysis, but is not necessary
for the asymptotic normality of the $t$-statistics that will be used.

\citet{Mackinnon2002} give a literature review and compare 14 different
methods for testing the effects of a mediation variable. These methods are
based on standardized measures of the product of two coefficients $\theta
_{1}\theta _{2}$ or based on the difference of two related coefficients ($%
\tau ^{\ast }-\tau $) in equations (\ref{eq:unrstricted_Y}) and (\ref%
{eq:restrictedmodel}):%
\begin{equation}
Y=\tau ^{\ast }X+u^{\ast }.  \label{eq:restrictedmodel}
\end{equation}%
If there is a mediation effect, then $X$ influences $M,$ such that $\theta
_{1}\neq 0,$ and $M$ influences $Y,$ such that $\theta _{2}\neq 0.$ If there
is no mediation by $M,$ then the effect of $X$ on $Y$ is not altered by the
inclusion of $M$ such that $\tau ^{\ast }-\tau =0.$

Model (\ref{eq:restrictedmodel}) is a restricted version of (\ref%
{eq:unrstricted_Y}) with $\theta _{2}=0,$ and it is straightforward
therefore to show that the OLS estimates for the three models satisfy $\hat{%
\tau}^{\ast }=\hat{\tau}+\hat{\theta}_{1}\hat{\theta}_{2}$ and the relation $%
\tau ^{\ast }-\tau =\theta _{1}\theta _{2}$ also holds in model
interpretation terms; see Appendix \ref{sec:AppendixTheory}.

\begin{figure}[h]
\par
\begin{center}
\includegraphics[width=0.55\textwidth]{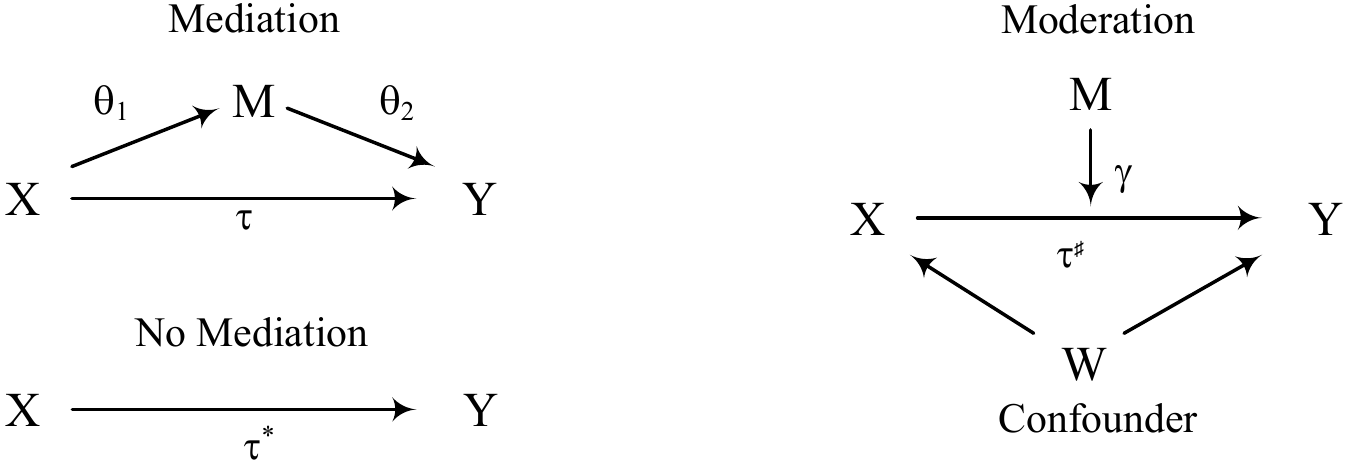}
\end{center}
\par
\vspace*{-6mm}
\caption{(No)-Mediation, Moderation, and Confoundedness}
\label{fig:MediationModerationConfoundedness}
\end{figure}

Figure \ref{fig:MediationModerationConfoundedness} illustrates mediation, no
mediation, and the related concept of moderation in terms of directed
acyclic graphs. $M$ is a moderator if it changes the relation between $X$
and $Y$. In its basic form this is modeled by adding the interaction term
between $X$ and $M$ to the model. When confounders $W$ are observed and
included we have:%
\begin{equation}
Y=\tau ^{\sharp }X+\gamma XM+\delta W+u^{\sharp }.
\end{equation}%
$W$ can also be included in (\ref{eq:unrstricted_Y}) - (\ref%
{eq:restrictedmodel}), but partialed out such that $Y$, $X,$ and $M$ are
residuals after regressions on $W.$ Moderation\ and confoundedness can be
tested by the ordinary $t$ or $F$ tests on $\gamma $ and $\delta ,$ but
testing for mediation is less straightforward.

The best well known and commonly used test for mediation by \citet{Sobel1982}
is a Wald-type test of the form $\ \hat{\theta}_{1}\hat{\theta}_{2}/SE(\hat{%
\theta}_{1}\hat{\theta}_{2})$, with $SE(\hat{\theta}_{1}\hat{\theta}_{2})$\
an estimate of the standard error of the product $\hat{\theta}_{1}\hat{\theta%
}_{2}.$ It is available in standard statistical packages such as SAS, R,
Stata, or SPSS. It has good properties when either $\theta _{1}$ or $\theta
_{2}$ is large and the standard errors of $\hat{\theta}_{1}$ and $\hat{\theta%
}_{2}$ are small, but if the two $t$-tests for testing $\theta _{1}=0$ and $%
\theta _{2}=0$ tend to be small, properties deteriorate. For parameter
values under the null, the Null Rejection Probability (NRP) can be very
close to zero and, under the alternative, power can fall far below the size
(highest NRP) of $5\%$ that we use throughout. All tests considered in %
\citet{Mackinnon2002} suffer from these problems. The distributions of the
test statistics considered in the literature depend on the value of the
parameters under the null. As a consequence none of these tests is similar,
meaning that rejection probabilities are not constant on the boundary of the
null hypothesis. In fact rejection probabilities under alternatives close to
the origin, i.e. power, can be much lower than size and these tests are
biased.

Much effort in the literature has gone into improving well-known test
statistics, such as the Wald statistic, without a satisfactory solution. The
bootstrap is invalid (see \citet{VanGarderenVanGiersbergen2021}) and cannot
salvage these statistics. The key step in our approach is to move away from
test statistics and consider the critical region in the sample space
directly and optimize the flexible boundary we define.

We make two main contributions. Our main theoretical contribution in Theorem %
\ref{Th:exactsimilartest} shows that a similar mediation test exists if and
only if the level of the test is the reciprocal of an integer, i.e. $%
1/\alpha \in 
\mathbb{N}
,$ or $\alpha =0$. Hence, for practical levels such as $1\%,$ $5\%$, or $%
10\% $ an exact similar test exists. The proof is constructive and the test
is unique within the class considered. Unfortunately, the critical region is
objectionable in that it includes an area near the origin where both $t$%
-statistics are arbitrarily close to zero. Such values do not provide
overwhelming evidence against the null and \citet{Perlman1999} coined the
term \textquotedblleft \emph{Emperor's New Tests}\textquotedblright\ for
similar tests with undesirable properties like these. Insistence on
similarity can render LR tests $\alpha $-inadmissible, cf. 
\citet[Section
6.7]{Lehmann2005}, but \citet{Perlman1999} give examples where similar tests
have extremely undesirable properties, yet inadmissible LR tests still
provide reasonable answers. In the mediation setting the LR test is also
inadmissible, but does not provide a satisfactory answer in this case. It is
much better than the Wald test as we will see, but nevertheless suffers from
extremely poor power properties for parameter values close to the origin.

So a better test is called for and our second, more important contribution
therefore is practical. We construct a new simple test for mediation that is
uniformly more powerful than the LR test without the undesirable properties
sketched by \citet{Perlman1999}, and that is nearly similar and practically
unbiased.

It is extremely simple to use:

\begin{enumerate}
\item Order the absolute values of the common $t$-statistics from the basic
OLS regressions (\ref{eq:unrstricted_Y}) and (\ref{eq:unrstricted_I}) for
testing $\theta _{1}=0$ or $\theta _{2}=0:$ $\left\vert t\right\vert
_{(1)}=\min \left\{ \left\vert t_{1}\right\vert ,\left\vert t_{2}\right\vert
\right\} $, $\left\vert t\right\vert _{(2)}=\max \left\{ \left\vert
t_{1}\right\vert ,\left\vert t_{2}\right\vert \right\} $

\item Reject if $\left\vert t\right\vert _{(1)}>g(\left\vert t\right\vert
_{(2)})$ using Table \ref{table:g_values_intro} (or employing the code in
Appendix \ref{sec:AppendixRCode}).
\end{enumerate}

\begin{table}[h]
\par
\begin{center}
\vspace*{2mm} {\footnotesize 
\begin{tabular}{|c|cccccccccc|}
\hline
$t$ & 0 & 0.01 & 0.02 & 0.03 & 0.04 & 0.05 & 0.06 & 0.07 & 0.08 & 0.09 \\ 
\hline
0.0 & 0 & 0.01 & 0.02 & 0.03 & 0.04 & 0.05 & 0.06 & 0.07 & 0.08 & 0.09 \\ 
0.1 & 0.1 & 0.10672 & 0.10672 & 0.10672 & 0.10672 & 0.10672 & 0.11672 & 
0.12671 & 0.13670 & 0.14669 \\ 
0.2 & 0.15669 & 0.16668 & 0.17667 & 0.18666 & 0.19666 & 0.20665 & 0.21664 & 
0.22663 & 0.23663 & 0.24662 \\ 
0.3 & 0.25661 & 0.26660 & 0.27660 & 0.28659 & 0.29658 & 0.30658 & 0.31657 & 
0.32656 & 0.33655 & 0.34655 \\ 
0.4 & 0.35654 & 0.36653 & 0.37652 & 0.38652 & 0.39651 & 0.40650 & 0.41649 & 
0.42649 & 0.43648 & 0.44647 \\ 
0.5 & 0.45646 & 0.46646 & 0.47645 & 0.48644 & 0.49643 & 0.50643 & 0.51642 & 
0.52641 & 0.53640 & 0.54640 \\ 
0.6 & 0.55639 & 0.56638 & 0.57637 & 0.58637 & 0.59636 & 0.60635 & 0.61634 & 
0.62634 & 0.63633 & 0.64632 \\ 
0.7 & 0.65631 & 0.66631 & 0.67630 & 0.68629 & 0.69628 & 0.70628 & 0.71627 & 
0.72626 & 0.73625 & 0.74625 \\ 
0.8 & 0.75624 & 0.76623 & 0.77622 & 0.78622 & 0.79621 & 0.80620 & 0.81620 & 
0.82619 & 0.83618 & 0.84617 \\ 
0.9 & 0.85617 & 0.86616 & 0.87615 & 0.88614 & 0.89614 & 0.90613 & 0.91612 & 
0.92611 & 0.93611 & 0.94610 \\ 
1.0 & 0.95609 & 0.96608 & 0.97608 & 0.98607 & 0.99606 & 1.00605 & 1.01605 & 
1.02604 & 1.03603 & 1.04602 \\ 
1.1 & 1.05602 & 1.06601 & 1.07600 & 1.08599 & 1.09599 & 1.10598 & 1.11597 & 
1.12596 & 1.13596 & 1.14595 \\ 
1.2 & 1.15594 & 1.16593 & 1.17593 & 1.18592 & 1.19591 & 1.20590 & 1.21590 & 
1.22589 & 1.23588 & 1.24587 \\ 
1.3 & 1.25587 & 1.26586 & 1.27585 & 1.28584 & 1.29584 & 1.30583 & 1.31286 & 
1.31310 & 1.31310 & 1.31310 \\ 
1.4 & 1.31310 & 1.31310 & 1.31310 & 1.31310 & 1.31310 & 1.31750 & 1.32750 & 
1.33750 & 1.34750 & 1.35750 \\ 
1.5 & 1.36750 & 1.37750 & 1.38750 & 1.39750 & 1.40750 & 1.41750 & 1.42750 & 
1.43750 & 1.44750 & 1.45750 \\ 
1.6 & 1.46750 & 1.47750 & 1.48750 & 1.49750 & 1.50750 & 1.51750 & 1.52750 & 
1.53750 & 1.54750 & 1.55750 \\ 
1.7 & 1.56750 & 1.57750 & 1.58750 & 1.59750 & 1.60750 & 1.61750 & 1.62750 & 
1.63750 & 1.64750 & 1.65750 \\ 
1.8 & 1.66750 & 1.67750 & 1.68750 & 1.69750 & 1.70750 & 1.71750 & 1.72750 & 
1.73750 & 1.74750 & 1.75750 \\ 
1.9 & 1.76750 & 1.77750 & 1.78750 & 1.79750 & 1.80750 & 1.81750 & 1.82750 & 
1.83750 & 1.84750 & 1.85750 \\ 
2.0 & 1.86750 & 1.87750 & 1.88750 & 1.89750 & 1.90750 & 1.91750 & 1.92750 & 
1.93750 & 1.94750 & 1.95750 \\ 
2.1 & 1.95996 & 1.95996 & 1.95996 & 1.95996 & 1.95996 & 1.95996 & 1.95996 & 
1.95996 & 1.95996 & 1.95996 \\ \hline
\end{tabular}
}
\end{center}
\par
\vspace*{-6mm}
\caption{$g$-function. Entries are $g(t)$ values for $t=$ value first
column+value first row. Reject if smallest absolute $t$-statistic is larger
than $g\left( \text{largest absolute }t\text{-statistic}\right) $ using
ordinary $t$-statistics from OLS regressions (\protect\ref{eq:unrstricted_Y}%
) and (\protect\ref{eq:unrstricted_I}). E.g. if $t_{1}=2.052$ and $%
t_{2}=-1.941$ then $1.941>1.9175=g(2.05)$ and reject. Linear interpolation
results in $\left\vert NRP-5\%\right\vert <0.00001.$}
\label{table:g_values_intro}
\end{table}

That the test can be based on elementary $t$-statistics is important for
ease of application, but is theoretically justified below by sufficiency and
invariance arguments. The testing problem is invariant to permutations of
parameters and statistics and sign changes. We show that the ordered
absolute $t$-statistic, $(\left\vert t\right\vert _{(1)},\left\vert
t\right\vert _{(2)}),$ is a maximal invariant. The critical region is a
subset of the relevant sample space which is an octant in $\mathbb{R}^{2}$.
This can also be justified asymptotically under very weak assumptions and
different estimation methods.

The new test is constructed by varying the boundary of the critical region,
defined as a function $g$, such that the test is almost similar and
minimizing the distance to the power envelope surface. It is based on a new
general, so called varying-$g$ method that can be applied to other testing
problems with nuisance parameters more generally to obtain near similar
tests. It does not require a choice of mixture distribution, nor the
construction of least favorable distributions, cf. \citet{Andrews1994}, %
\citet{Andrews2006,Andrews2008}, \citet{Elliott2015}, \cite{Guggenberger2019}%
. It can be given a random critical value interpretation as in %
\citet{moreira2016critical} since any critical region in a
higher-dimensional space has a boundary for one statistic in terms of the
remaining statistics. The critical region that we construct is fixed
however, not random, avoids simulation, and our approach appears to lend
itself better to multivariate extensions, as will be shown for dimension
three.

We use and develop numerical methods that avoid simulations and use
numerical integration instead. With the required computing completed for the
mediation problem, practitioners can simply use Table \ref%
{table:g_values_intro} or the computer code provided in the appendix. In
fact no further reading on the motivation and derivation of the test is
required for its implementation.

Section \ref{sec:TeachApplication} shows the ease of implementation using an
interesting application by \citet{Alan2018} on educational attainment. The
test confirms that the negative effect on girls' attainment of 1-year
exposure to teachers with traditional attitudes, is mediated through
students' own gender role beliefs. Neither the procedure in \citet{Alan2018}
nor the LR test reject the no mediation hypothesis in this case, but our new
test does.

A further empirical illustration on union sentiment among southern nonunion
textile workers is provided in Section \ref{sec:EmpiricalIllustration}.
Different mediation channels are tested involving two mediating variables
and requires an extension of our methods. We therefore consider general
hypotheses of the form $H_{0}:\theta _{1}\cdots \theta _{K}=0$. We give the
relevant distributions of maximal invariants that can be used to derive the
critical regions that are nearly similar and do so explicitly for three
dimensions.

For practitioners the major advantage of our test is that there is a better
chance of formally showing that there is a mediation effect. Our test has
better power, especially when the two channeling effects are small or less
accurately estimated. Given the enormous interest in testing for mediation
and the fact that our test can have close to $5\%$ more power than standard$%
\ $level $5\%$ tests, many unpublished examples will exist where it can now
be concluded that there is a statistically significant mediation effect.

\section{Theory}

\label{sec:Theory2}

The joint density of $\left( Y,M\right) $ given $X$ in equations (\ref%
{eq:unrstricted_Y}) and (\ref{eq:unrstricted_I}) can be written as: $\ $%
\begin{equation*}
f_{Y,M|X}\left( y,m|x;\lambda \right) =f_{Y|M,X}\left( y|m,x;\lambda
_{1}\right) f_{M|X}\left( m|x;\lambda _{2}\right) ,
\end{equation*}%
with $\lambda _{1}=\left( \tau ,\theta _{2},\sigma _{11}\right) ^{\prime }$, 
$\lambda _{2}=\left( \theta _{1},\sigma _{22}\right) ^{\prime }$, and $%
\lambda =\left( \lambda _{1}^{\prime },\lambda _{2}^{\prime }\right)
^{\prime }$. The parameters $\lambda _{1}$, and $\lambda _{2}$ vary freely
as a result of the triangular structure of the model. The mediation variable
is the endogenous variable in (\ref{eq:unrstricted_I}), but is exogenous for 
$\theta _{2}$ in (\ref{eq:unrstricted_Y}) since $Y$ is not causal for $M$.
For a sample of $n$ independent observations the loglikelihood equals the
sum of two normal loglikelihoods corresponding to (\ref{eq:unrstricted_Y})
and (\ref{eq:unrstricted_I})\footnote{%
This can easily be extended to include more regressors/covariates.
Instrumental variables can also be used, but note that $X$ and $M$ appear in
both equations and in the standard setup $u$ and $v$ are independent because
of the experimental interpretation of $M$. All that is required is the
asymptotic normality of the estimators and $t$-statistics.}:

\begin{equation}
\ell \left( \lambda \right) \propto -\frac{1}{2\sigma _{11}}%
\sum_{i=1}^{n}\left( y_{i}-\tau x_{i}-\theta _{2}m_{i}\right) ^{2}-\frac{n}{2%
}\log \left( \sigma _{11}\right) -\frac{1}{2\sigma _{22}}\sum_{i=1}^{n}%
\left( m_{i}-\theta _{1}x_{i}\right) ^{2}-\frac{n}{2}\log \left( \sigma
_{22}\right) .  \label{eq:loglik}
\end{equation}%
As a consequence the Maximum Likelihood Estimators (MLEs) for $\theta _{1}$
and $\theta _{2}$ are the basic OLS\ estimators for the two equations
separately. Furthermore, both observed and expected Fisher information
matrices will be block diagonal in terms of $\lambda _{1}$ and $\lambda _{2}$
as well as in $\left( \tau ,\theta _{2}\right) ^{\prime }$, $\sigma _{11}$, $%
\theta _{1},$ and $\sigma _{22}.$ As a result the standard $t$-statistics $%
T_{1}$ and $T_{2}$ for $\theta _{1}$ and $\theta _{2}$ respectively are
asymptotically independent and normally distributed with means $\mu
_{1}\equiv \theta _{1}^{0}/\sigma _{\theta _{1}}$, $\mu _{2}\equiv \theta
_{2}^{0}/\sigma _{\theta _{2}}$ where $\theta _{1}^{0}$, $\theta _{2}^{0}$
denote the true parameter values and $\sigma _{\theta _{1}}$, $\sigma
_{\theta _{2}}$ the standard deviations of the OLS estimators: $\ \left(
T-\mu \right) \overset{d}{\rightarrow }\emph{N}\left( 0,I_{2}\right) $, with 
$T=\left( T_{1},T_{2}\right) ^{\prime }.$

Restricting attention to $T$ can be justified by statistical sufficiency,
since the MLE\ is minimal sufficient and complete, and by invariance since
the values of $\sigma _{11},\sigma _{22}$ and $\tau $ do not affect whether $%
H_{0}:$ $\theta _{1}\theta _{2}=0$ is true or not. %
\citet{HillierVanGarderenVanGiersbergen2021} show that $T$ is a maximal
invariant under a relevant group of transformations. The testing problem has
two more obvious symmetries. The problem is not affected by sign changes or
permutations. This also holds in higher dimensions, e.g. when mediation is
through a chain of effects as in our empirical illustration. If $%
X\rightarrow M^{\left( 0\right) }\rightarrow \cdots \rightarrow M^{\left(
K-1\right) }\rightarrow Y,$ \ then $K$ parameters are required to be
non-zero for this channel to operate. In $K$ dimensions the null hypothesis
that at least one parameter is zero and the alternative is all $K$
parameters non-zero:%
\begin{align*}
H_{0}& :\theta _{1}\theta _{2}...\theta _{K}=0 \\
H_{1}& :\theta _{1}\theta _{2}...\theta _{K}\neq 0
\end{align*}%
There are many other hypotheses including collapsibility in contingency
tables, testing indirect effects, channels in DAGs, see \citet{Dufour2017}
for a range of examples.

In the multivariate setting we also assume that the estimator $\hat{\theta}$
is normally distributed with known covariance matrix $\Omega =diag\left(
\sigma _{1}^{2},...,\sigma _{K}^{2}\right) $. Hence, if we let $T=\Omega
^{-1/2}\hat{\theta}$ and $\mu =\Omega ^{-1/2}\theta ^{0}$ such that $T_{k}=%
\hat{\theta}_{k}/\sigma _{\theta _{k}}$ are the $t$-ratios and $\mu
_{k}=\theta _{k}/\sigma _{\theta _{k}}$ are the non-centrality parameters,
we assume:

\begin{assumption}
\label{As:Gaussianity}:$\qquad T\sim N\left( \mu ,I_{K}\right) .$
\end{assumption}

The testing problem is invariant to reordering the parameters (permutations)
and sign changes (reflections) of the $K$ parameters $\left\{ \theta
_{i}\right\} _{i=1}^{K}$. The group of permutations, $\mathbf{G}_{1}$ say,
has $K!$ elements and the group of sign changes, $\mathbf{G}_{2}$ say, has $%
2^{K}$ elements (two possible signs for each element). The groups $\mathbf{G}%
_{1}$ and $\mathbf{G}_{2}$ have only the identity element in common, but are
otherwise non-overlapping. The full group $\mathbf{G=G}_{1}\mathbf{\times G}%
_{2}$ generated by $\mathbf{G}_{1}$ and $\mathbf{G}_{2}$ therefore has $%
K!2^{K}$ elements. Exploiting the invariance and symmetry properties of the
problem reduces the domain of integration by a factor $K!2^{K}$. This is
important because all optimizations require probabilities calculated by
numerical integration. The density after a sign change in $T_{k}$ is
obtained by a corresponding sign change in $\mu _{k}$ and for a permutation
of $T$ also $\mu $ permutes accordingly. Hence for any element \textbf{$g$}$%
\in \mathbf{G}$ we have $\mathbf{g}\cdot T\sim $ $\emph{N}\left( \mathbf{g}%
\cdot \mu ,I_{K}\right) $ or $P_{\mathbf{g}\mu }\left[ \mathbf{g}T\in A%
\right] =P_{\mu }\left[ T\in A\right] $ so the distribution is invariant;
see \citet{Lehmann2005}.

\begin{theorem}
\label{Th:MaxInvAbsOrderStatK}The testing problem $H_{0}:\mu _{1}\mu
_{2}\cdots \mu _{K}=0$ is invariant under the group of transformations \ $%
\mathbf{G=G}_{1}\mathbf{\times G}_{2}$ acting on $T$ and $\mu $, given
Assumption \ref{As:Gaussianity}.\ The absolute order statistic $\left(
\left\vert T\right\vert _{\left( 1\right) },...,\left\vert T\right\vert
_{\left( K\right) }\right) $ with $0<\left\vert T\right\vert _{\left(
1\right) }<\left\vert T\right\vert _{\left( 2\right) }<...<\left\vert
T\right\vert _{\left( K\right) }$ is a maximal invariant statistic and the
absolute order parameter $\left( \left\vert \mu \right\vert _{\left(
1\right) },...,\left\vert \mu \right\vert _{\left( K\right) }\right) $ with $%
0\leq \left\vert \mu \right\vert _{\left( 1\right) }\leq \cdots \leq
\left\vert \mu \right\vert _{\left( K\right) }$ is a maximal invariant
parameter under the group of transformations\textbf{\ }$\mathbf{G}=\mathbf{G}%
_{1}\times \mathbf{G}_{2}.$ The distribution of $\left( \left\vert
T\right\vert _{\left( 1\right) },...,\left\vert T\right\vert _{\left(
K\right) }\right) $ depends only on $\left( \left\vert \mu \right\vert
_{\left( 1\right) },...,\left\vert \mu \right\vert _{\left( K\right)
}\right) .$
\end{theorem}

The Wald and LR tests are functions of the maximal invariant and so will our
new test. The density is required for probability calculations and
optimizations for the new test. It is easily derived for arbitrary dimension
using Equation (6) of \citet{Vaughan1972}:

\begin{lemma}
\label{Lem:MaxInvAbsOrderStatKdistribution}The probability density function
of the absolute order statistic is given by:%
\begin{equation}
f_{\left\{ \left\vert T\right\vert _{\left( 1\right) },...,\left\vert
T\right\vert _{\left( K\right) }\right\} }\left( \left\vert t\right\vert
_{\left( 1\right) },...,\left\vert t\right\vert _{\left( K\right) }\right)
=perm\left( 
\begin{array}{ccc}
\chi \left( \left\vert t\right\vert _{\left( 1\right) },\left\vert \mu
\right\vert _{\left( 1\right) }\right) & \cdots & \chi \left( \left\vert
t\right\vert _{\left( 1\right) },\left\vert \mu \right\vert _{\left(
K\right) }\right) \\ 
\vdots &  & \vdots \\ 
\chi \left( \left\vert t\right\vert _{\left( K\right) },\left\vert \mu
\right\vert _{\left( 1\right) }\right) & \cdots & \chi \left( \left\vert
t\right\vert _{\left( K\right) },\left\vert \mu \right\vert _{\left(
K\right) }\right)%
\end{array}%
\right) ,  \label{eq:perm-density}
\end{equation}%
with $perm\left( A\right) $ the permanent\footnote{%
The permanent is defined as $perm\left( A\right) =\sum_{\sigma \in
S_{n}}\prod_{i=1}^{n}a_{i,\sigma \left( i\right) }$ with the sum over all
permutations $\sigma $ of the numbers $1,...,n,$ akin the determinant but
without the $\pm $ signature of the permutation.} of the square matrix $A$
and $\chi \left( x,\mu \right) $ the noncentral Chi-distribution with one
degree of freedom and noncentrality parameter $\mu >0$.
\end{lemma}

The noncentral Chi-distribution $\chi \left( t,\mu \right) $ with one degree
of freedom equals the folded normal distribution and if $T\sim N\left( \mu
,1\right) ,$ then the density of $\left\vert T\right\vert $ can be written
as $f_{|T|}(t,\mu )=\,\sqrt{2/\pi }\exp \{-\tfrac{1}{2}(t^{2}+\mu
^{2})\}\cosh (\mu \ t),$ for $t\geq 0$. Substitution in Lemma \ref%
{Lem:MaxInvAbsOrderStatKdistribution} and simplifying gives the following
result that is the basis for the numerical calculations that follow:

\begin{lemma}
\label{Lem:pdf_2_order_t_statistics}The density of the ordered absolute $t$%
-statistics for the mediation hypothesis is:%
\begin{eqnarray*}
f_{\left\{ \left\vert T\right\vert _{\left( 1\right) },\left\vert
T\right\vert _{\left( 2\right) }\right\} }\left( t_{1},t_{2};\mu _{1},\mu
_{2}\right) &=&\frac{2}{\pi }\exp \left\{ -(t_{1}^{2}+t_{2}^{2}+\mu
_{1}^{2}+\mu _{1}^{2})/2\right\} \\
&&\left\{ \cosh (\mu _{1}t_{1})\cosh (\mu _{2}t_{2})+\cosh (\mu
_{1}t_{2})\cosh (\mu _{2}t_{1})\right\} , \\
&&\text{ \ \ for }t_{2}\geq t_{1}\geq 0~\text{and }\mu _{1}=\theta
_{1}/\sigma _{\theta _{1}},\mu _{2}=\theta _{2}/\sigma _{\theta _{2}}.
\end{eqnarray*}
\end{lemma}

The ordered squared $t$-statistic could also be used as maximal invariant.
Lemma \ref{Lem:MaxInvAbsOrderStatKdistribution} would then lead to a density
in terms of noncentral Chi-squared distributions.

\subsection{Problems with Standard (Single) Mediation Test Statistics}

Standard mediation test statistics used in practice have distributions that
depend on the parameter values under the null. The rejection probabilities
are therefore not constant and the tests are biased with power dropping
below the size of the test, especially in a neighborhood of the origin. We
illustrate the issue for the classic Wald and LR tests.

The null hypothesis $\theta _{1}\theta _{2}=0$ defines a manifold that is
almost everywhere continuously differentiable, with the exception of the
origin which is a so-called \textquotedblleft double point" where the two
restrictions $\theta _{1}=0$ and $\theta _{2}=0$, each defining a
one-dimensional line, coincide. The widely used \citet{Sobel1982} test
equals the square root of the Wald test. \citet{Glonek1993} derives the
asymptotic distribution for the Wald test statistic:%
\begin{equation}
W=\frac{T_{1}^{2}T_{2}^{2}}{T_{1}^{2}+T_{2}^{2}}\overset{d}{\rightarrow }%
\left\{ 
\begin{array}{l}
\ \chi _{1}^{2}:\text{if }\theta _{1}=0\text{ or }\theta _{2}=0,\text{ but
not both,} \\ 
\ \frac{1}{4}\chi _{1}^{2}:\text{if }\theta _{1}=\theta _{2}=0.%
\end{array}%
\right.  \label{eq:Wald}
\end{equation}%
As a consequence the asymptotic $5\%$ critical value for $W$ when both $%
\theta _{1}=0$ and $\theta _{2}=0$ is $\frac{1}{4}\chi _{1}^{2}(0.95),$ but
jumps to $\chi _{1}^{2}(0.95)$ i.e. the usual Chi-squared critical value for
one restriction, for any other value. The discrete jump in the asymptotic
distribution from the origin to any other fixed parameter is remarkable and
shows explicitly that the distribution depends heavily on the parameter
values under the null. This discontinuity in the asymptotic distribution\
and dependence on the parameter also invalidates bootstrap procedures and
they are oversized. For an NRP of $5\%$ at the origin the critical value
should be $0.96$ but this would lead to over-rejection for other values
under the null and the test would be oversized (size $>30\%$) and invalid.
One could consider drifting sequences of parameter values to investigate the
behavior of the Wald statistic near the origin, but that would not solve the
problem. The problematic behavior of the Wald test under the null with
singularities is well documented by \citet{DrtonXiao2016} and %
\citet{drton2009likelihood}. No satisfactory solution has been found in the
preceding decades to salvage the Wald statistic, see e.g. \citet{Dufour2017}%
. This prompted our investigation and to propose an alternative solution.

The LR test was shown by \citet{vanGiersbergen2014}\ to equal:%
\begin{equation}
LR=\min \left\{ \left\vert T_{1}\right\vert ,\left\vert T_{2}\right\vert
\right\} =\left\vert T\right\vert _{\left( 1\right) },  \label{eq:LR}
\end{equation}%
and rejects when both $H_{0}^{\theta _{1}}:\theta _{1}=0$ and $H_{0}^{\theta
_{2}}:\theta _{2}=0$ are rejected by basic $t$-tests. In %
\citet{Mackinnon2002} this is referred to as the test for joint
significance, but not identified as the LR test. The rejection probability
for critical value $cv$ is: 
\begin{equation*}
P\left[ LR\geq cv\right] =P\left[ \ \left\vert T_{1}\right\vert \geq cv\
\cap \ \left\vert T_{2}\right\vert \geq cv\right] =P\left[ \ \left\vert
T_{1}\right\vert \geq cv\right] \ \cdot P\left[ \ \left\vert
T_{2}\right\vert \geq cv\right] ,
\end{equation*}%
by independence of $T_{1}$ and $T_{2}$. These rejection probabilities are
monotonically increasing in the absolute values of $\theta _{1}$ and $\theta
_{2}$. Correct size is therefore obtained by choosing the critical value of
the test by letting $\theta _{1}\rightarrow \infty $ when $\theta _{2}=0,$
or $\theta _{2}\rightarrow \infty $ if $\theta _{1}=0,$ to guarantee that
the rejection probability under the null is always smaller than or equal to
the nominal size. The asymptotic $5\%$ critical value is therefore the usual 
$1.96$. The NRP will depend on the values of $\theta _{1}$ and $\theta _{2}$
and vary between the following two extremes: 
\begin{equation*}
P\left[ LR\geq z_{0.025}\right] =\left\{ 
\begin{array}{l}
0.05:\text{ if }\theta _{1}\rightarrow \infty \text{ }\wedge \theta _{2}=0,%
\text{or }\theta _{2}\rightarrow \infty \wedge \theta _{1}=0, \\ 
0.05^{2}=0.0025:\text{ if }\theta _{1}=0\wedge \theta _{2}=0,%
\end{array}%
\right.
\end{equation*}%
where $z_{0.025}$ is the upper $2.5\%$ percentile of the standard normal
distribution. For an NRP of $5\%$ at the origin $\left( \theta _{1},\theta
_{2}\right) =\left( 0,0\right) ,$ the critical value should equal $%
cv_{LR00}=1.217$. This leads to massive over-rejection if only one parameter
is zero and the other much larger. The test with this critical value is
oversized $($size $>20\%)$ and invalid.

The third classic test, the Lagrange Multiplier (LM) or score test, is even
more problematic because its definition depends on parameter values under
the null and there are three different versions depending on which $\theta $%
, or both $\theta $'s are zero.

All these classic tests are functions of two $t$-statistics. Their
distributions, as well as their NRPs, clearly depend on the parameter values
under the null and the tests are not similar. A test is called \emph{similar
on the boundary} of $H_{0}$ if the probability of rejecting the null is
constant for all parameter values on the boundary of $H_{0}$ and $H_{1}$.
For mediation, this boundary equals $H_{0}$ itself and consists of the
horizontal and vertical axes of the $\left( \theta _{1},\theta _{2}\right) $
space. None of the classic tests is similar and in a neighborhood of the
origin the NRPs are close to zero. As a result the power in a neighborhood
of the origin is also close to zero and far below the size of the test and
the tests are biased since there are parameter values with probability of
rejection under the alternative lower than under the null.

\subsection{Critical Regions}

The behavior and construction of the classic test statistics is problematic.
Given that no satisfactory adjustments of classic test statistics have been
found, despite considerable efforts over recent decades, a different
approach is required.

In order to derive an alternative test procedure we shift the focus from the 
\emph{test statistic} to the \emph{critical region} (CR). A critical region
defines a test statistic of course, but choosing a class of tests, such as
Wald, LR, or LM tests, restricts the shape of the critical region. For the
same reason the tests focusing on improving the standard error of $\hat{%
\theta}_{1}\hat{\theta}_{2}$ or $\left( \hat{\tau}^{\ast }-\hat{\tau}\right) 
$ analyzed in \citet{Mackinnon2002} restrict possible shapes of the critial
region.

We construct a new test procedure by constructing the critical region
directly by determining its shape in the two-dimensional sample space of the 
$t$-statistics used in the construction of the tests. We consider critical
regions that are bounded by a function $g$ and reject when $\left\vert
T\right\vert _{\left( 1\right) }>g(\left\vert T\right\vert _{\left( 2\right)
}).$ We impose some weak regularity conditions. In particular we assume that 
$g$ is\ a c\`{a}dl\`{a}g function from $%
\mathbb{R}
_{0}^{+}$ (including 0) to $%
\mathbb{R}
_{0}^{+}.$ This will allow $g$ to have jumps, but limits the number of jumps
to countably many. We also insist that $g$ is weakly increasing. This
assures that a rejection (acceptance) for a realization $(\left\vert
t\right\vert _{\left( 1\right) },\left\vert t\right\vert _{\left( 2\right)
}) $ is not reversed when either $\left\vert t\right\vert _{\left( 1\right)
} $ or $\left\vert t\right\vert _{\left( 2\right) }$ is increased
(decreased). This reasoning also motivates a CR that is\ (topologically)
simply connected. Denote the set of weakly increasing c\`{a}dl\`{a}g
functions by $\mathbb{D}\left( 
\mathbb{R}
_{0}^{+},%
\mathbb{R}
_{0}^{+}\right) $, as in the\ common Skorokhod space notation, but here with
weak monotonicity.

\begin{definition}
A function $g\in \mathbb{D}\left( 
\mathbb{R}
_{0}^{+},%
\mathbb{R}
_{0}^{+}\right) $ is called the \emph{boundary function} (of the critical
region) when it defines:%
\begin{align*}
\text{Critical Region }& :CR_{g}=\left\{ \left( T_{1},T_{2}\right) \in 
\mathbb{R}
^{2}\mid \left\vert T\right\vert _{\left( 1\right) }>g(\left\vert
T\right\vert _{\left( 2\right) })\right\} , \\
\text{Acceptance Region }& :AR_{g}=\left\{ \left( T_{1},T_{2}\right) \in 
\mathbb{R}
^{2}\mid \left\vert T\right\vert _{\left( 1\right) }\leq g(\left\vert
T\right\vert _{\left( 2\right) })\right\} .
\end{align*}
\end{definition}

Justification for only using $t$-statistics is by sufficiency and
invariance. First, the MLE \ \ $\hat{\lambda}=\left( \hat{\tau},\hat{\theta}%
_{2},\hat{\sigma}_{11},\hat{\theta}_{1},\hat{\sigma}_{22}\right) ^{\prime }$
\ is a complete minimal sufficient statistic. The model constitutes a full
exponential model since the dimensions of the minimal sufficient statistic
and the parameter space are equal; see \citet{vanGarderen1997}. Second, $%
T_{1}$ and $T_{2}$ have distributions under the null that are independent of
the nuisance parameters $\tau ,$ $\sigma _{11},$ and $\sigma _{22}.$ %
\citet{HillierVanGarderenVanGiersbergen2021} shows that, also in finite
samples, $T=\left( T_{1},T_{2}\right) ^{\prime }$ is a maximal invariant
under an appropriate group of transformations generalizing the scale
invariance of the $t$-statistics. Theorem \ref{Th:MaxInvAbsOrderStatK} shows
that as a consequence of the permutation and reflection invariance only $%
1/8^{th}$ of the two-dimensional sample space of $T$ needs to be considered
and we define the critical region in the first octant (east to north-east).
The other seven octants in $\mathbb{R}^{2}$ follow by symmetry. The test
defined by $CR_{g}$ is indeed invariant to permutations, reflections, and
scale transformations. The domain of $g\left( \cdot \right) $ can therefore
be restricted to the non-negative real line and bounded by the $45^{\circ }$
line: $g\left( x\right) \leq x.$

We can put the definition of a similar test in terms of the boundary
function $g\left( \cdot \right) $, noting that $H_{0}$ is itself the
boundary of $H_{0}$ and $H_{1}$:

\begin{definition}
$g\left( \cdot \right) $ is said to be a \emph{similar }boundary function if
the probability of the critical region $CR_{g}$ defined by $g$ is constant
under $H_{0}$: 
\begin{equation*}
P\left[ T\in CR_{g}\mid \forall \left( \theta _{1},\theta _{2}\right) \in 
\mathbb{R}^{2}\ with\ \theta _{1}\theta _{2}=0 \right] =constant.\ 
\end{equation*}
\end{definition}

The boundary functions defined by the Wald (Sobel) and LR are not similar.
Figures \ref{fig:CRWalds} and \ref{fig:CRLRs} show the boundary functions
that define the critical region in terms of $\left( T_{1},T_{2}\right) $ for
the Wald and LR test. We show the boundaries for two critical values: one
such that for large $\left\vert \theta _{1}\right\vert $ or $\left\vert
\theta _{2}\right\vert $ the NRP is $5\%$ asymptotically. This value is the
usual $3.84$ for the Wald test and $1.96$ for the LR\ test. The second,
smaller critical value is such that the NRP is $5\%$ when $\theta
_{1}=\theta _{2}=0.$ This value is $0.96$ for the Wald test and $1.217$ for
the LR test. The rejection probabilities are shown as a function of the
noncentrality parameter $\mu _{2}=\theta _{2}/\sigma _{\theta _{2}}$ for
given $\mu _{1}=0,$ such that $H_{0}$ holds. For the LR test with critical
value $1.96$ the NRP goes down to $0.0025=0.05^{2}$ when $\theta _{1}=\theta
_{2}=0.$ In the second case, with the smaller critical value $1.217$, the
NRP is $5\%$ by construction when $\theta _{1}=\theta _{2}=0,$ but this is
not a valid test since for other values the NRPs are much higher than the
nominal size of $5\%$. The same situation will occur when constructing point
optimal invariant tests. The Wald test is considerably worse than the LR
test with lower NRP over a wider range of $\mu _{1}$, see figures \ref%
{fig:NRPWald} and \ref{fig:NRPLR}. 
\begin{figure}[h]
\centering
\begin{tabular}{cc}
\begin{subfigure}{0.4\textwidth} \centering
\includegraphics[width=0.95\textwidth]{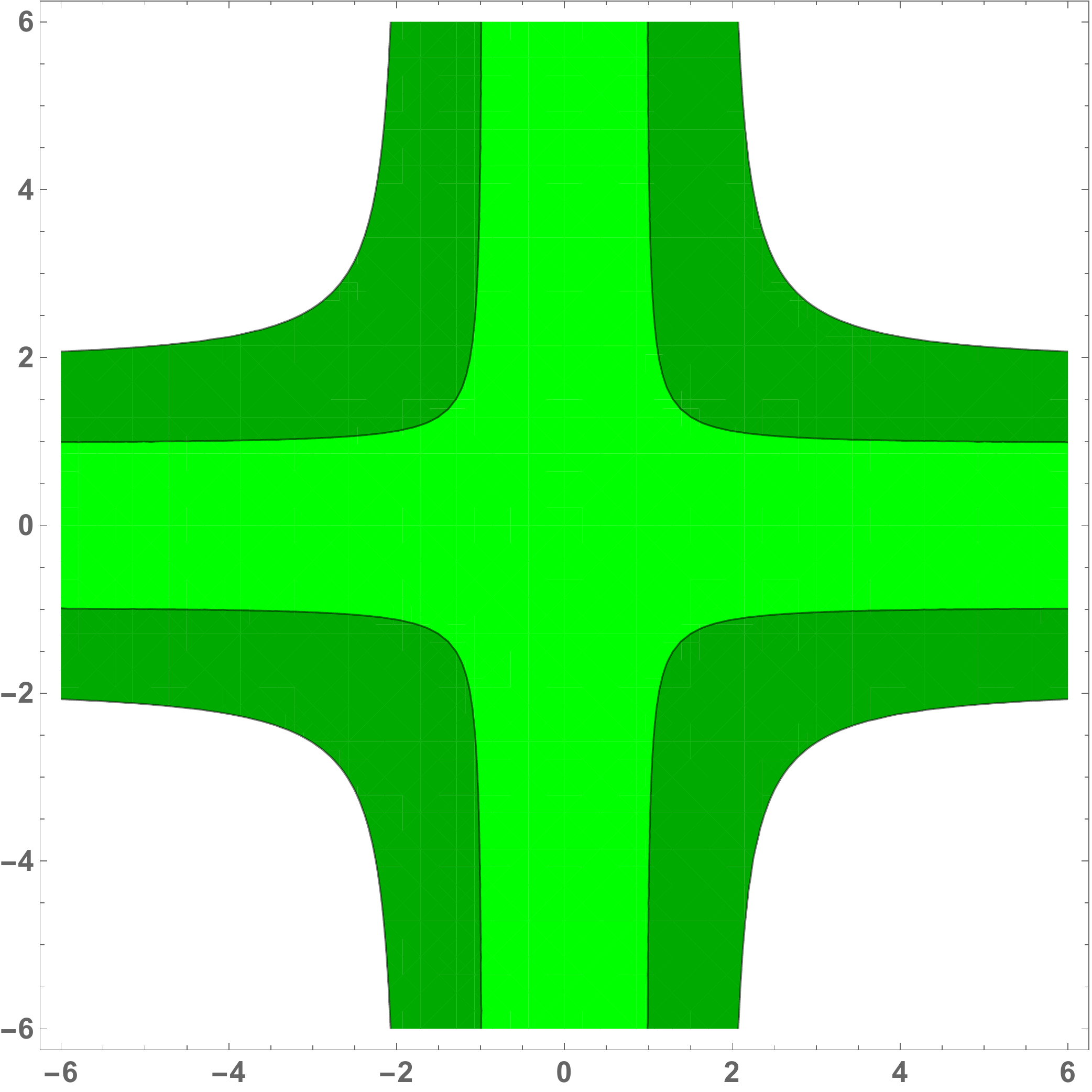} \caption{CR Wald
(Sobel) test. } \label{fig:CRWalds} \end{subfigure} & \begin{subfigure}{0.4%
\textwidth} \centering \includegraphics[width=0.95\textwidth]{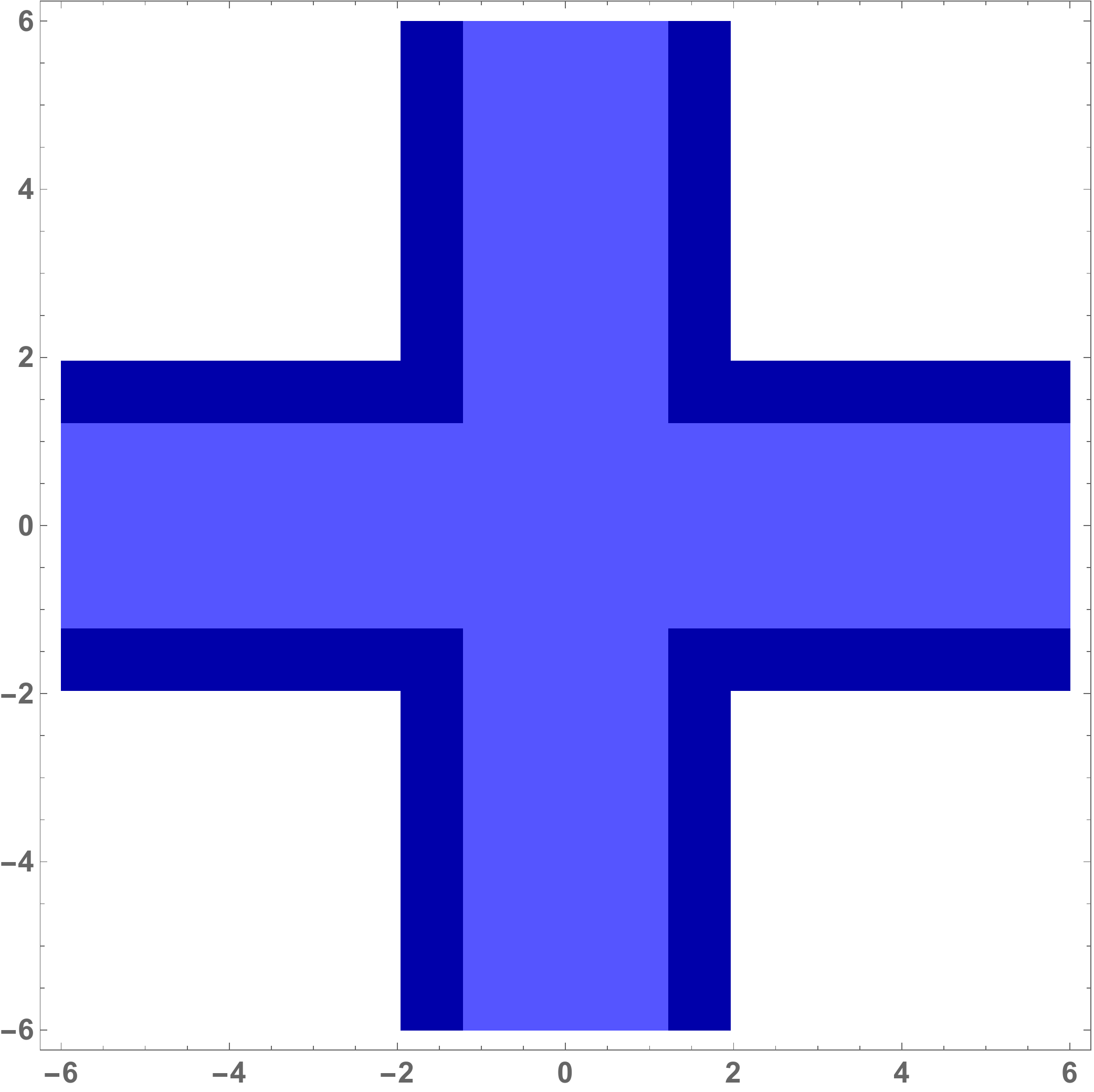}
\caption{CR LR test. } \label{fig:CRLRs} \end{subfigure} \\ 
\begin{subfigure}{.4\textwidth} \centering
\includegraphics[width=0.95\textwidth]{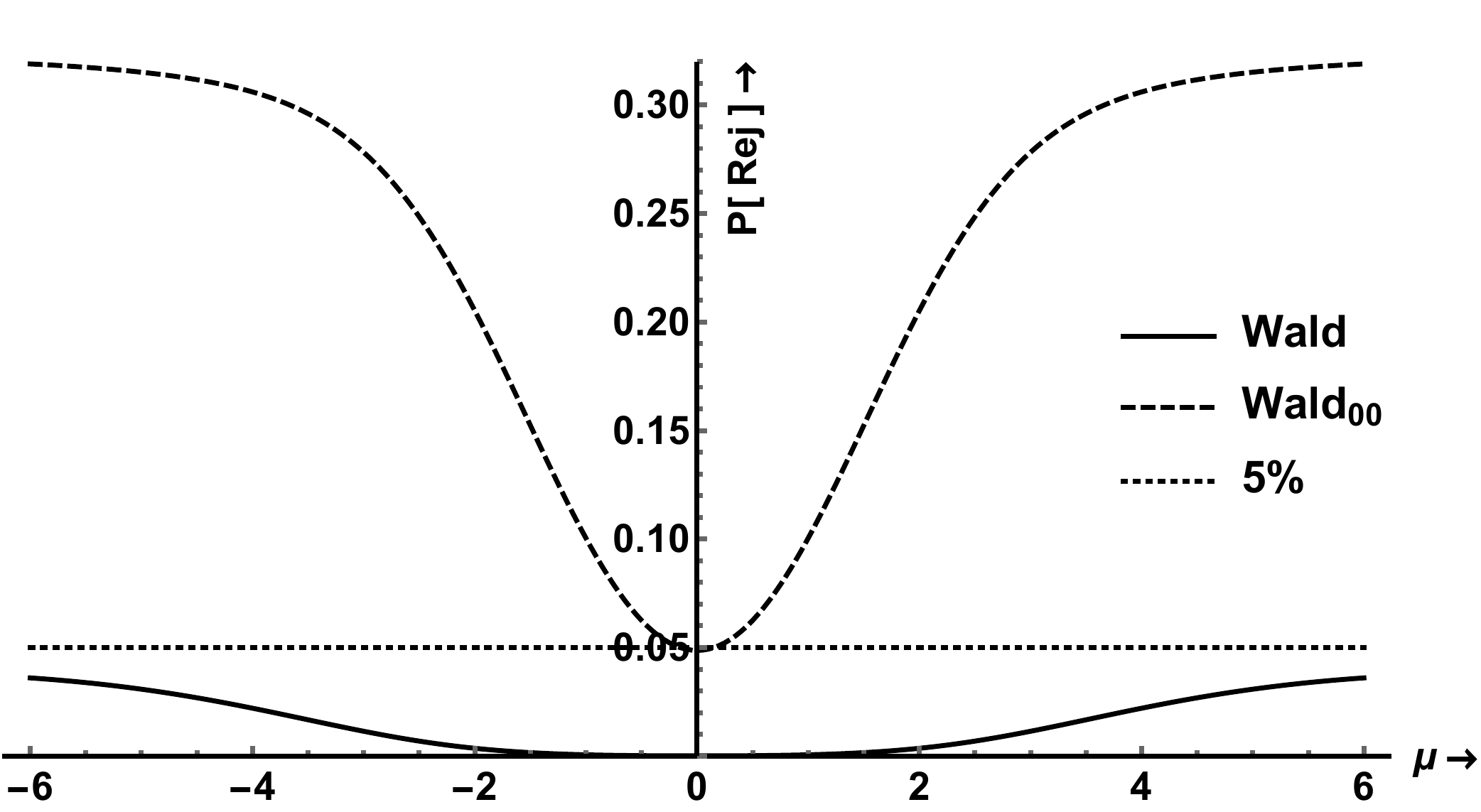} \caption{NRP Wald
(Sobel) test.} \label{fig:NRPWald} \end{subfigure} & \begin{subfigure}{.4%
\textwidth} \centering
\includegraphics[width=0.95\textwidth]{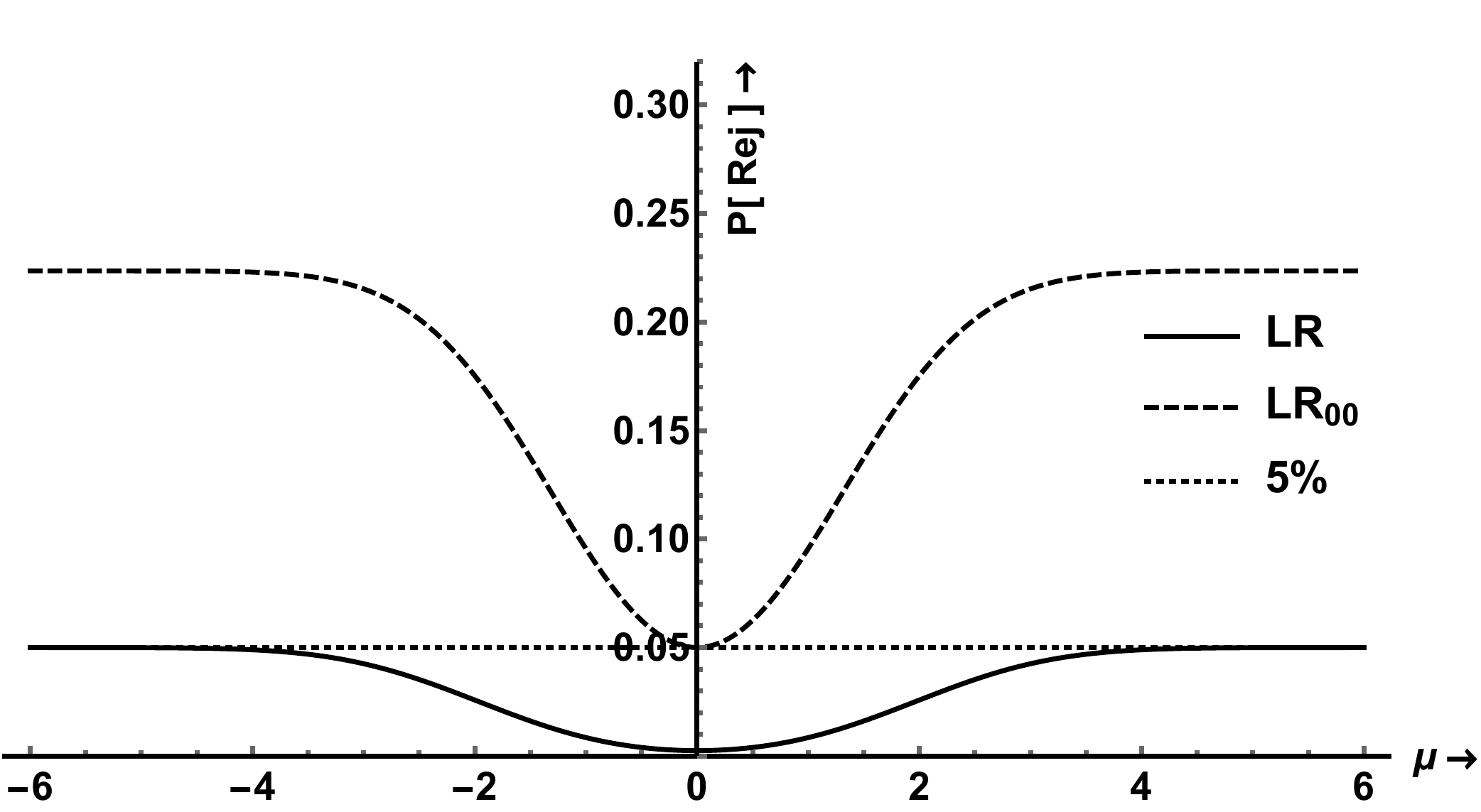} \caption{NRP LR
test.} \label{fig:NRPLR} \end{subfigure} \\ 
& 
\end{tabular}%
\vspace*{-6mm}
\caption[short]{Critical regions for Sobel (Wald) and LR tests and their
rejection probabilities. White areas are the critical regions for the valid $%
5\%$ tests. For $5\%$ rejection in the origin $(\protect\alpha ,\protect%
\beta )=(0,0)$, the dark areas are added to the critical regions which leads
to over-sized tests. The LR test is poor, but the Sobel test is even worse.}
\end{figure}

The trinity of classic tests is clearly nonsimilar. The question is if we
can do much better. Does there exist a similar test or is this a problem
that is intrinsically unsolvable? The next theorem, and the main theoretical
contribution, answers this question.

\begin{theorem}
\label{Th:exactsimilartest}A similar boundary function $g\left( \cdot
\right) $ exists for testing $H_{0}:\theta _{1}\theta _{2}=0$ if and only if 
$1/\alpha $ is an integer (or trivially $\alpha =0$). If it exists, the
boundary is unique in $\mathbb{D}\left( 
\mathbb{R}
_{0}^{+},%
\mathbb{R}
_{0}^{+}\right) .$
\end{theorem}

For common significance levels, including $1\%$, $5\%$, and $10\%$, the
theorem proves that exact similar tests exist. The proof is given in
Appendix \ref{sec:AppendixProofNonExistenceTheorem} and exploits the
symmetries of the problem and the completeness of the normal distribution.
The proof is constructive showing that the $g$ function must be a step
function and is unique within the class of weakly increasing c\`{a}dl\`{a}g
functions. For $\alpha =0.25$ the exact similar boundary is shown in Figure %
\ref{Fig:ExactSimBoundalfa25}.

\begin{figure}[h]
\par
\begin{center}
\includegraphics[width=0.45\textwidth]{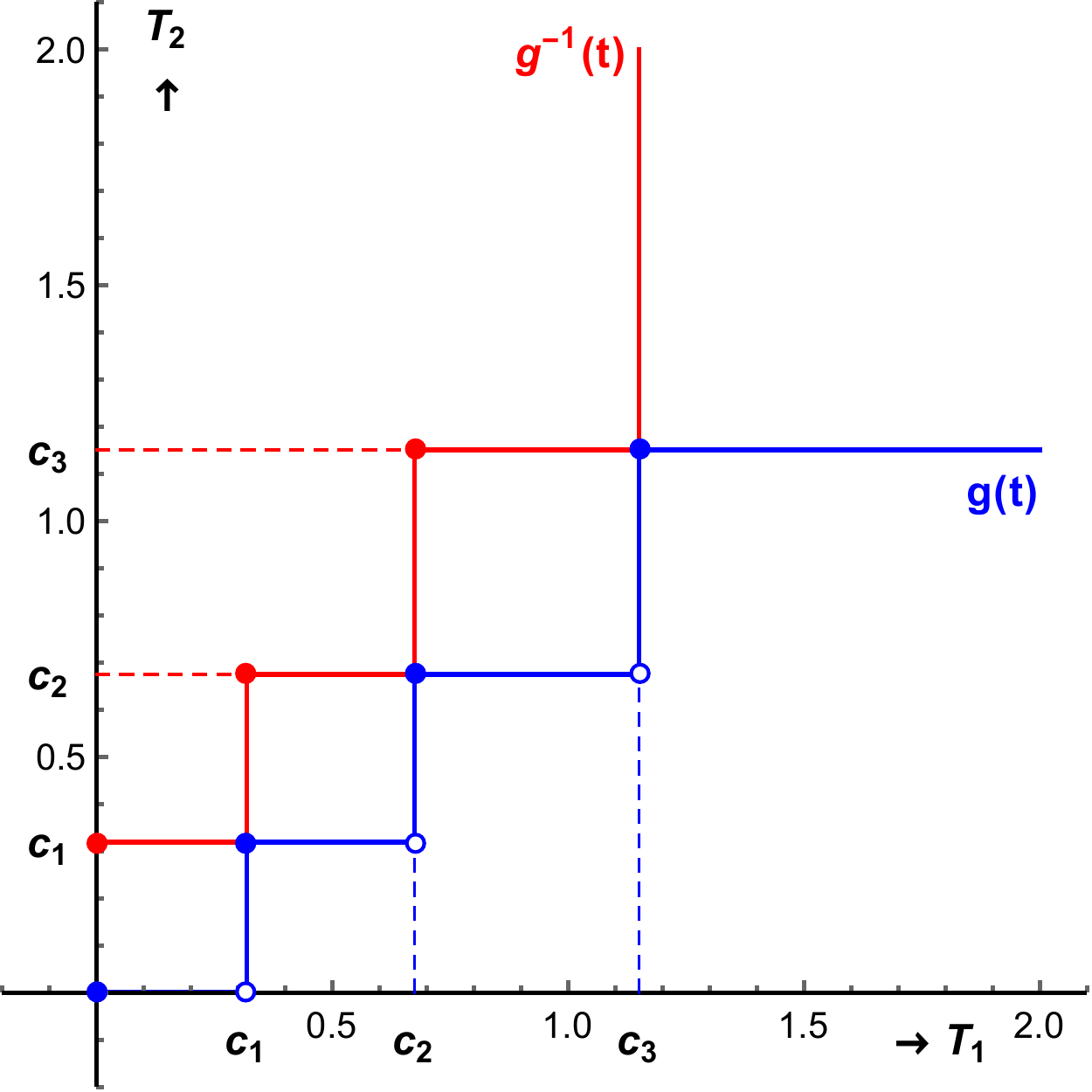}
\end{center}
\par
\vspace*{-6mm}
\caption[shot]{Boundary of the exact similar $g$-test when $\protect\alpha %
=0.25$ in the first quadrant. }
\label{Fig:ExactSimBoundalfa25}
\end{figure}

The CR shows a clear parallel with Figure 2 of \citet{berger1989}. The step
function starts horizontal and defines a CR with objectionable features. In
particular, the CR will include all $T$ such that $0<\left\vert T\right\vert
_{\left( 2\right) }<\Phi ^{-1}\left( \frac{1}{2}+\frac{\alpha }{2}\right) $.
For $\alpha =0.05$ this corresponds to rejection when both $t$-statistics
are smaller than $0.0627$ in absolute value (but non-zero) \ (for $\alpha
=0.10$ and $0.01$ smaller than $0.1257$ and $0.0125$ respectively). Such $t$%
-statistics close to zero cannot be characterized as strong evidence against
the null, in favor of both $\theta _{1}$ and $\theta _{2}$ being non-zero.
Nevertheless, the test is size correct and renders the LR test $\alpha $%
-inadmissible. So it is an \textquotedblleft \emph{Emperor's New Tests}%
\textquotedblright\ in the terminology of \citet{Perlman1999} and does not
provide a satisfactory solution to the problem of finding a better test. A
second unattractive feature of this exact similar boundary is that for
increasing values of the test statistics parallel and close to the diagonal,
the decision alternates between rejection and acceptance, despite the fact
that evidence against the null appears monotonically increasing.

Given the uniqueness of the similar test, we relax the strict similarity
requirement and consider a class of near similar tests with NRPs that differ
from $\alpha $ by no more than $\epsilon $. Within this class of so called $%
\epsilon $-similar tests given in Definition \ref{Def:NearSimilarG} below,
we determine a test that avoids the objectionable properties of the exact
similar test, but achieves good power properties. Within this class of near
similar tests we determine the power envelope and determine the test that
minimizes the distance between its power surface and this power envelope.
This new test is therefore optimal in this sense. It is easy to implement
using Table \ref{table:g_values_intro} or using the R-code in Appendix \ref%
{sec:AppendixRCode}.

The first step in the composition of this optimal test is a new general
method for constructing near similar tests.

\section{Near Similar Test Construction: Varying \textit{g}-Method}

A general method for constructing near similar tests involves three generic
steps:

\begin{enumerate}
\item Define a flexible boundary $g$ for the critical region in the relevant
sample space.

\item Define a criterion function $Q\left( g\right) $ that penalizes the
deviation of the NRP from the level $\alpha $ for a grid of parameter values
under the null (and possibly restrictions on $g$ and other aspects deemed
relevant).

\item Systematically vary and determine $g$ such that it minimizes the
criterion function and is therefore as close to similarity as possible in
the metric defined by $Q$.
\end{enumerate}

The relevant sample space is determined by the particular testing problem at
hand and may have been reduced by sufficiency, invariance, or other
principles, to dimension $k$, say. The boundary $g$ of the critical and
acceptance region is then of dimension $\left( k-1\right) $, but may consist
of disjoint parts if the critical and/or acceptance region are not simply
connected in a topological sense. There are various possibilities to define $%
g$ flexibly, but we will use splines.

The criterion function may include aspects other than similarity, for
instance smoothness and monotonicity of $g$, convexity of the critical or
acceptance regions, or even rejection probabilities under alternatives.
Consequently, Step 3 will generally be a constraint optimization problem.
The systematic variation of $g$ is intended to be in line with the
optimization routine used to minimize $Q$ as in, e.g. a Newton-Raphson-type
procedure.

For mediation testing with $\alpha =0.05$ an exact test exists, but the
varying $g$-method may not find it because (i) $g$ is not flexible enough
(ii) $Q\left( g\right) $ includes criteria other than NRPs (iii) numerical
difficulties determining the jumps (iv) and finally restrictions to
purposely exclude objectionable CRs.

An explicit implementation of the varying $g$-method to the mediation
problem is given next.

\subsection{Near Similar Mediation \textit{g}-Test}

Step 1 in the varying $g$-method is to determine the relevant sample space
for the testing problem. The first dimensional reduction is to the MLE which
is minimal sufficient and complete. The second reduction to $\left(
T_{1},T_{2}\right) $ follows from location-scale type invariance shown in %
\citet{HillierVanGarderenVanGiersbergen2021}. Permutation and reflection
symmetries further reduce the relevant sample space to one octant according
to Theorem \ref{Th:MaxInvAbsOrderStatK}. The maximal invariant is the
ordered absolute $t$-statistic $(\left\vert T\right\vert _{\left( 1\right)
},\left\vert T\right\vert _{\left( 2\right) })=\left( \min \left( \left\vert
T_{1}\right\vert ,\left\vert T_{2}\right\vert \right) ,\max \left(
\left\vert T_{1}\right\vert ,\left\vert T_{2}\right\vert \right) \right) $
with density given in Lemma \ \ref{Lem:pdf_2_order_t_statistics} . For
general $\mu $, the rejection probability (RP) for the $g$-method is given
by:%
\begin{equation}
\pi _{g}\left( \mu _{1},\mu _{2}\right) =1-\int_{0}^{\infty
}\int_{0}^{g(t_{2})}f_{|T|_{(1)},|T|_{(2)}}(t_{1},t_{2},\mu _{1},\mu
_{2})dt_{1}dt_{2}.  \label{RP}
\end{equation}%
This two-dimensional integral can be reduced to a one-dimensional integral
by expressing the inner integral in terms of the CDF of the standard normal
distribution $\Phi \left( \cdot \right) $:%
\begin{eqnarray}
\int_{0}^{b}f_{|T|_{(1)},|T|_{(2)}}(t_{1},t_{2},\mu _{1},\mu _{2})dt_{1}\,
&=&\sqrt{\frac{2}{\pi }}\exp (-t_{2}^{2}/2)\cdot  \notag \\
&&\{\exp (-\mu _{2}^{2}/2)\cosh (t_{2}\mu _{2})\left( \Phi (b-\mu _{1})+\Phi
(b+\mu _{1})-1\right) +  \notag \\
&&\exp (-\mu _{1}^{2}/2)\cosh (t_{2}\mu _{1})\left( \Phi (b-\mu _{2})+\Phi
(b+\mu _{2})-1\right) \}.  \label{F1f2}
\end{eqnarray}%
The RP thus simplifies to a one-dimensional integral by integrating formula (%
\ref{F1f2}) for $b=g(t_{2})$ over $t_{2}\in \lbrack 0,\infty )$, which
greatly improves efficiency and accuracy.

All NRPs in the paper were determined by numerical integration over $%
t_{2}\in \lbrack 0,\mu _{2}+9]$ \ under the null with $\mu _{1}=0$. The
numerical integration was performed in \emph{Julia}, see %
\citet{Bezanson17julia}, using the function \emph{quadgk} that is based on
adaptive Gauss-Kronrod quadrature. For the final $g$-function, the estimated
upper bound on the absolute error for the calculated NRP is approximately $%
10^{-8}$.

The $g$-boundary is generally determined by an algorithm and Appendix \ref%
{sec:AppendiixAlgorithms} shows the basic implementation of the varying $g$%
-method using linear splines with $J+2$ knots, with the first and last knots
fixed. In spite of its simplicity, it leads to big improvements even for
small values of $J$. Figure \ref{fig:CRgconstructJ6} illustrates the
construction of the $g$-function for a fixed number of grid points $J=6$ and
the resulting $CR_{g}$ in the sample space of $( \left\vert T\right\vert
_{\left( 1\right) },\left\vert T\right\vert _{\left( 2\right) }) $, the East
to North-East octant of $%
\mathbb{R}
^{2}$.

\begin{figure}[h]
\par
\begin{center}
\includegraphics[width=0.5\textwidth]{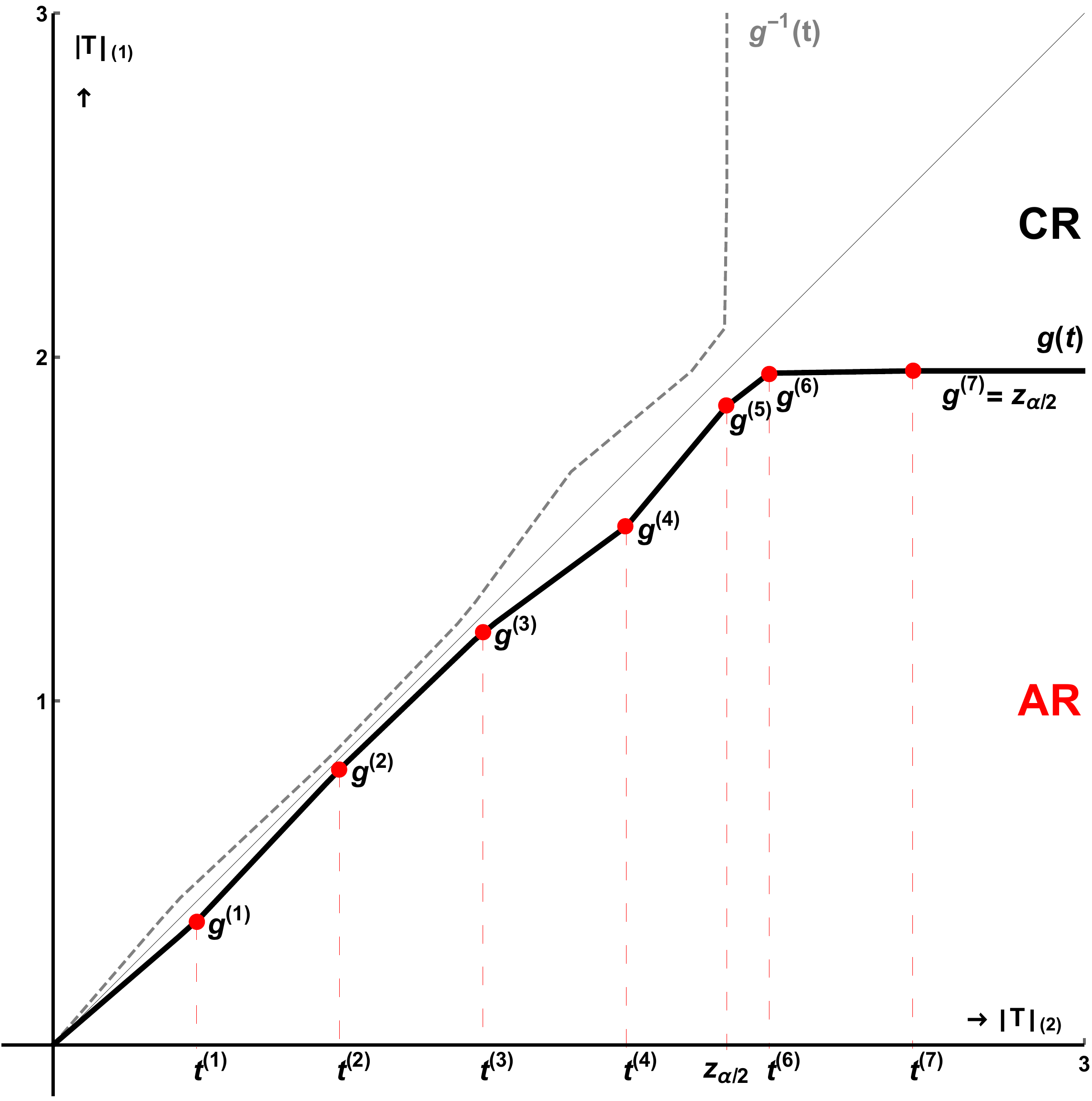}
\end{center}
\par
\vspace*{-6mm}
\caption{Construction of the basic $g$-function with $J=6$, so $8$ knots in
all and the resulting CR boundary in $( \left\vert T\right\vert _{\left(
2\right) },\left\vert T\right\vert _{\left( 1\right) }) $ space.}
\label{fig:CRgconstructJ6}
\end{figure}

Figure \ref{fig:NRPcompareJ369LRW} shows the NRPs of the test in comparison
to the LR\ and Wald (Sobel) tests. There is a remarkable gain in the lowest
NRP, and therefore local power, from $0.25\%$ to $4.999\%$. Already for $J=2$
there is a large improvement and after $J=8$ gains are small, and beyond $%
J=16$ there was hardly any improvement. 
\begin{figure}[h]
\par
\begin{center}
\includegraphics[width=0.6\textwidth]{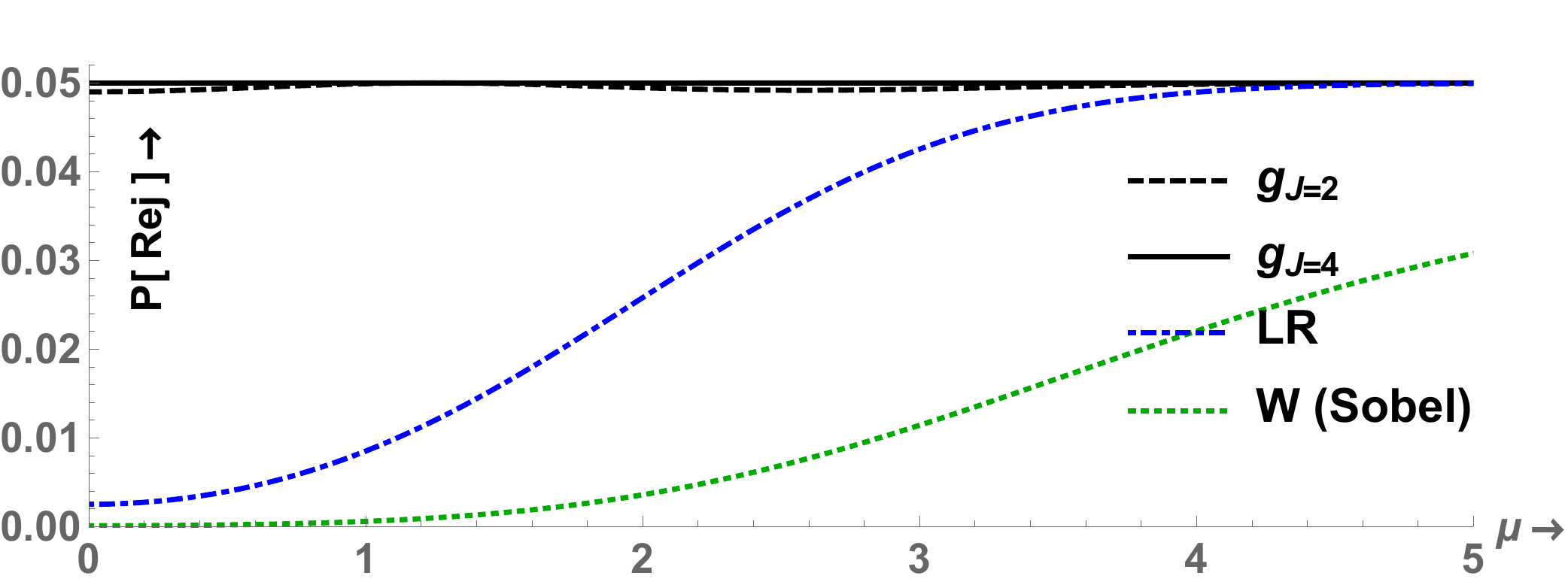}
\end{center}
\par
\vspace*{-6mm}
\caption{NRPs $g$-tests with $J=2,4$ versus LR and Wald (Sobel) tests.}
\label{fig:NRPcompareJ369LRW}
\end{figure}

\subsection{Power}

Insistence on similarity can have negative consequences for the power in
general, but not here. The power envelope with or without (near) similarity
restriction are very close. The power surface of our optimal test is also
close. Even the basic test with $J=16$ has very good power in comparison
with the Sobel and LR tests and is uniformly better for all values of the
noncentrality parameter $\mu $. In a neighborhood of the origin with $\mu
=0, $ it is essentially $5\%$ points higher.

The RP $\pi _{g}\left( \mu _{1},\mu _{2}\right) $ defined in Equation (\ref%
{RP}) is the NRP if $\mu _{1}$ and/or $\mu _{2}$ equal $0$ (the null
hypothesis). When both are non-zero,\ $H_{0}$ is false and $\pi _{g}\left(
\mu _{1},\mu _{2}\right) $ is the power of the test defined by $CR_{g}$.
Figure \ref{fig:PowerComparisonLJ369LRW} illustrates the power in the $%
45^{\circ },$ $\mu _{1}=\mu _{2},$ direction. Power in other directions is
also superior to the Wald (Sobel) and LR tests.

\begin{figure}[h]
\par
\begin{center}
\includegraphics[width=0.60\textwidth]{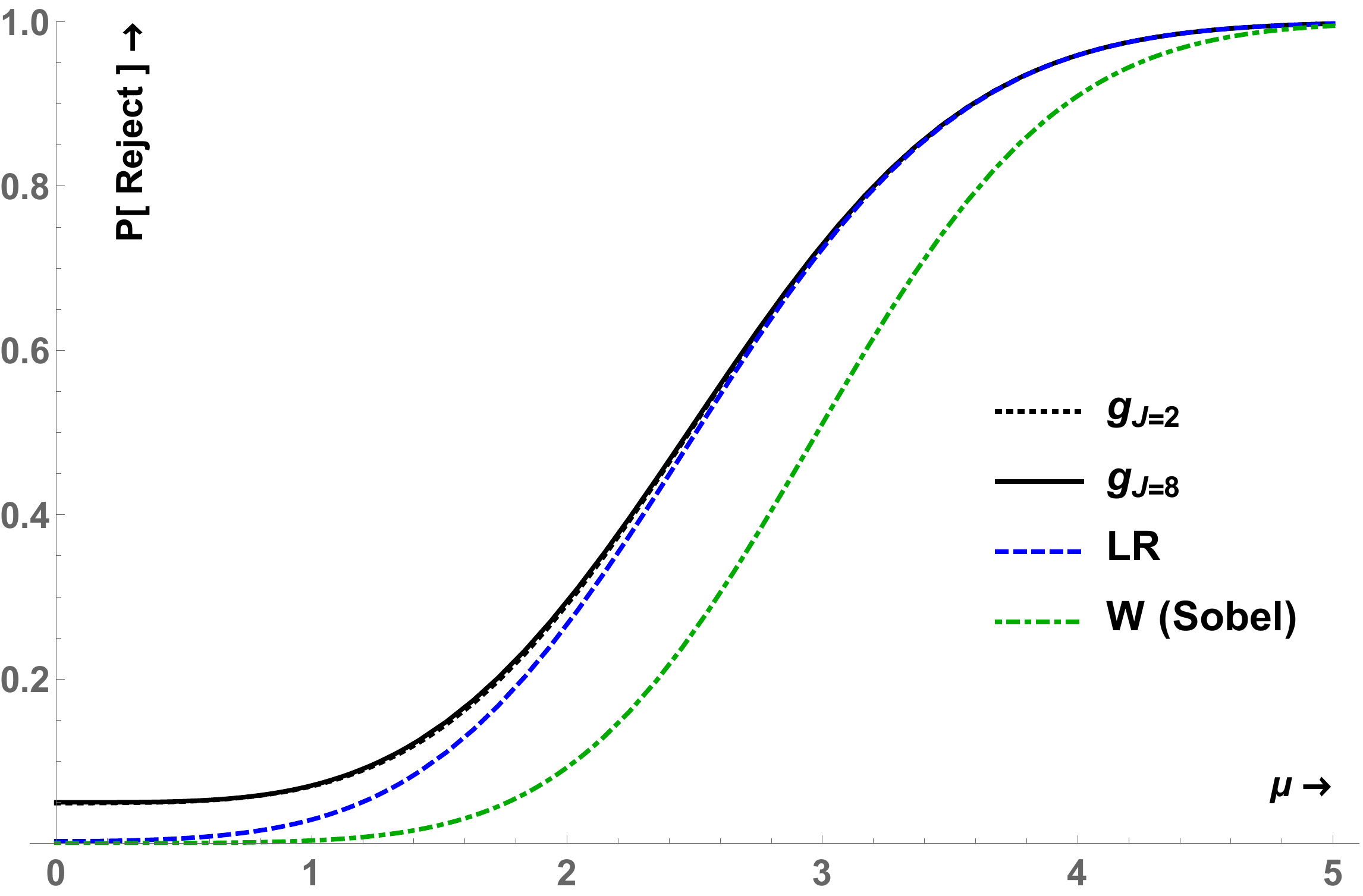}
\end{center}
\par
\vspace*{-6mm}
\caption{Power comparison basic $g$-tests, LR- and W (Sobel) tests along $%
\protect\mu _{1}=\protect\mu _{2} \ \ (=\protect\mu )$.}
\label{fig:PowerComparisonLJ369LRW}
\end{figure}
\medskip


There is a straightforward explanation for the additional power. The Wald
and LR\ test both reject far less than $5\%$ near the origin. The critical
region can therefore be extended and the power increased without failing the
size condition. In the origin the NRPs are close to $0\%$ for the LR\ and
Wald (Sobel) tests. By extending the critical region we can therefore gain
almost $5\%$ power without violating the size condition. The LR test has
some attractive features including that rejection for a particular $\left(
t_{1},t_{2}\right) $ implies rejection for larger values of $t_{1}$ and/or $%
t_{2}$ also. This is intuitive since the evidence against the null is
increasing. A disadvantage is however that one never rejects when either $%
t_{1}$ or $t_{2}$ is smaller than $1.96$ and this causes extreme
conservativeness that can be resolved by adding area to the critical region.
It should also be noted that the relevant distribution of the order
statistic $(\left\vert T\right\vert _{\left( 1\right) },\left\vert
T\right\vert _{\left( 2\right) })$ is quite different from the distribution
of the product of two standard normals for $\left( T_{1},T_{2}\right) $ or
their absolute values.

\subsection{Power Envelope}

Comparison to the Sobel and LR tests is limited because they both perform
very poorly for small values of $\mu $. The absolute quality, or even near
optimality, of the new $g$-test can only be assessed by comparing the power
surface of the test to the power envelope, or a tight upper bound thereof,
for a class of tests that satisfy appropriate invariance-, size and almost
similarity restrictions. Since the exact similar invariant test is
objectionable when it exists, we introduce a class $\Gamma _{\alpha
,\epsilon }$ of near similar tests with NRPs that deviate less than $%
\epsilon $ from the $\alpha $ level and an operational (super)class $\Gamma
_{\alpha ,\epsilon }^{\mathbb{M}_{0}}\supseteq \Gamma _{\alpha ,\epsilon }$
as follows.

\begin{definition}
\label{Def:NearSimilarG}The class $\Gamma _{\alpha ,\epsilon }$ of near
similar boundary functions with $\epsilon >0$ and given $\alpha $ is defined
by: 
\begin{equation*}
\Gamma _{\alpha ,\epsilon }=\left\{ g\in \mathbb{D}\left( 
\mathbb{R}
_{0}^{+},%
\mathbb{R}
_{0}^{+}\right) \left\vert \sup_{\mu _{0}\geq 0}P\left[ CR_{g}\mid \left(
0,\mu _{0}\right) \right] \leq \alpha \text{ and }\inf_{\mu _{0}\geq 0}P%
\left[ CR_{g}\mid \left( 0,\mu _{0}\right) \right] \geq \alpha -\epsilon
\right. \right\}
\end{equation*}%
The class $\Gamma _{\alpha ,\epsilon }^{\mathbb{M}_{0}}$ with set $\mathbb{M}%
_{0}=\left\{ \left( 0,\mu _{0}^{\left( \iota \right) }\right) \right\}
_{\iota =1}^{\Upsilon _{0}}$ containing $\Upsilon _{0}$ points under $H_{0}$%
, is defined by: 
\begin{equation*}
\Gamma _{\alpha ,\epsilon }^{\mathbb{M}_{0}}=\left\{ g\in \mathbb{D}\left( 
\mathbb{R}
_{0}^{+},%
\mathbb{R}
_{0}^{+}\right) \left\vert \sup_{\left( 0,\mu _{0}\right) \in \mathbb{M}%
_{0}}P\left[ CR_{g}\mid \left( 0,\mu _{0}\right) \right] \leq \alpha \text{
and }\inf_{\left( 0,\mu _{0}\right) \in \mathbb{M}_{0}}P\left[ CR_{g}\mid
\left( 0,\mu _{0}\right) \right] \geq \alpha -\epsilon \right. \right\} .
\end{equation*}
\end{definition}

For $\epsilon =0$ the boundary functions in $\Gamma _{0}$ would be similar.
We consider $\alpha =0.05$ in all our numerical examples, but for general $%
1/\alpha \notin 
\mathbb{N}
$ the set would be empty according to Theorem \ref{Th:exactsimilartest}. For 
$\epsilon =\alpha $, on the other hand, $\Gamma _{\alpha ,\alpha }$ contains 
\emph{all} $g$-based tests that satisfy the size condition. Our interest is
in $\epsilon $ close to $0$, when $\Gamma _{\alpha ,\epsilon }$ contains
boundaries that are almost similar. The minimum value of $\epsilon $ for
which $\Gamma _{\alpha ,\epsilon }$ is not empty is 0 for the mediation
problem when $1/\alpha \in 
\mathbb{N}
$, but in general depends on the testing problem considered and larger than $%
0$ if no similar test exists.

The class $\Gamma _{\alpha ,\epsilon }^{\mathbb{M}_{0}}$ can be thought of
as a discretization of $\Gamma _{\alpha ,\epsilon }$ in the sense that a
grid of points under the null is considered. It imposes less restrictions
and enforces near similarity conditions on a finite number of points only.
As a consequence it may contain boundaries that do not satisfy the size
condition for points that are not in $\mathbb{M}_{0}.$ Obviously $\Gamma
_{\alpha ,\epsilon }$\ $\subseteq \Gamma _{\alpha ,\epsilon }^{\mathbb{M}%
_{0}}$ since the size and NRP conditions also hold for the points in $%
\mathbb{M}_{0}$.

Within the class $\Gamma _{\alpha ,\epsilon }$ there is no unique solution.
As a consequence one has to choose a boundary function from $\Gamma _{\alpha
,\epsilon },$ or in practice from $\Gamma _{\alpha ,\epsilon }^{\mathbb{M}%
_{0}},$ to obtain an operational test. For the construction of the power
envelope we can select the test in $\Gamma _{\alpha ,\epsilon }^{\mathbb{M}%
_{0}}$ that maximizes the power against a particular point $\left( \mu
_{1},\mu _{2}\right) $ in the alternative. This test is a Point Optimal
Invariant Near Similar (POINS) test. The critical region of this test varies
with $\left( \mu _{1},\mu _{2}\right) $ and no uniformly most powerful test
exists within the class $\Gamma _{\alpha ,\epsilon }$. It can be used
however, to construct an upper bound for the power envelope.

\begin{definition}
The power envelope of a near similar invariant test with $\epsilon >0$ is
defined as: 
\begin{equation*}
\pi \left( \mu _{1},\mu _{2}\right) =\max_{g\in \Gamma _{\alpha ,\epsilon
}}P \left[ CR_{g}\mid \left( \mu _{1},\mu _{2}\right) \right] .
\end{equation*}%
For a given set of points $\mathbb{M}_{0}=\left\{ (0,\mu _{0}^{\left( \iota
\right) })\right\} _{\iota =1}^{\Upsilon _{0}}$ an upper bound to the power
envelope is:%
\begin{equation*}
\bar{\pi}\left( \mu _{1},\mu _{2}\right) =\max_{g\in \Gamma _{\alpha
,\epsilon }^{\mathbb{M}_{0}}}P\left[ CR_{g}\mid \left( \mu _{1},\mu
_{2}\right) \right] .
\end{equation*}
\end{definition}

For notational simplicity we have suppressed the dependence on $\epsilon $
and $\mathbb{M}_{0}.$ Since $\Gamma _{\alpha ,\epsilon }$\ $\subseteq \Gamma
_{\alpha ,\epsilon }^{\mathbb{M}_{0}}$ and elements of $\Gamma _{\alpha
,\epsilon }^{\mathbb{M}_{0}}$ do not necessarily satisfy the size condition
for all parameter values it follows that $\bar{\pi}\left( \mu _{1},\mu
_{2}\right) \geq \pi \left( \mu _{1},\mu _{2}\right) ,$ because fewer
conditions are imposed.\ Choosing a finer grid for $\mathbb{M}_{0}$ will
force $\bar{\pi}\left( \mu _{1},\mu _{2}\right) \ $closer to $\pi \left( \mu
_{1},\mu _{2}\right) $, at least in the additional points in $\mathbb{M}_{0}$
where the size condition is now required to hold. Also note that the
\textquotedblleft point\textquotedblright\ optimal $g$ that maximizes power
for the point $\left( \mu _{1},\mu _{2}\right) ,$ may have undesirable
features such as including parts of the axes in the critical region, even
though such observations are perfectly in line with the null hypothesis.

We determine $\bar{\pi}\left( \mu _{1},\mu _{2}\right) $ numerically for $%
\alpha =0.05$ by maximizing the power directly by selecting critical region
points in the sample space that maximize the probability of rejection when
the true density has parameter $\left( \mu _{1},\mu _{2}\right) $, under the
side conditions that the NRP$\in \left[ 0.05-\epsilon ,0.05\right] $ for all
parameters $\left( 0,\mu _{0}\right) \in \mathbb{M}_{0}$. The sample space
is decomposed into $285,150$ squares\ and a modern optimization routine is
used to determine which squares should be included in the critical or
acceptance region in order to maximize the power while at the same time
satisfying the approximate similarity condition. This is repeated for a grid
of $\left( \mu _{1},\mu _{2}\right) $ points. So for each point on the grid
the POINS\ critical region is determined and the power recorded. Appendix %
\ref{sec:AppendiixAlgorithms} gives details of the algorithm and the
optimization routine that can deal with a large number of variables and side
conditions.

By dropping the near similarity restriction ($0.05-\epsilon \leq NRP)$ in
the same algorithm, we can construct a power envelope for nonsimilar tests.
The maximal difference from the (higher) nonsimilar power surface is $2\%$
points when power is around $40\%,$ showing that the power loss due to the
similarity requirement is small. The calculated $\bar{\pi}\left( \mu
_{1},\mu _{2}\right) $ surface enables us to construct a correctly sized
optimal test derived in the next section.

\section{The New Mediation Test}

\label{sec:OptimalgTest}

Having determined an upper bound to the power envelope, we can determine a $%
g $-boundary function with a power surface as close as possible to this
upper bound. This optimal test is found using the algorithm given in
Appendix \ref{sec:AppendiixAlgorithms}. This function is given in Appendix %
\ref{sec:AppendixRCode}~and R-code is also provided there. For ease of
implementation we give values of $g\left( t\right) $ in Table \ref%
{table:g_values_intro}. Figure \ref{fig:OptimalVG} shows the optimal $g$%
-boundary test for the mediation problem. 
\begin{figure}[h]
\centering
\includegraphics[width=0.6\textwidth]{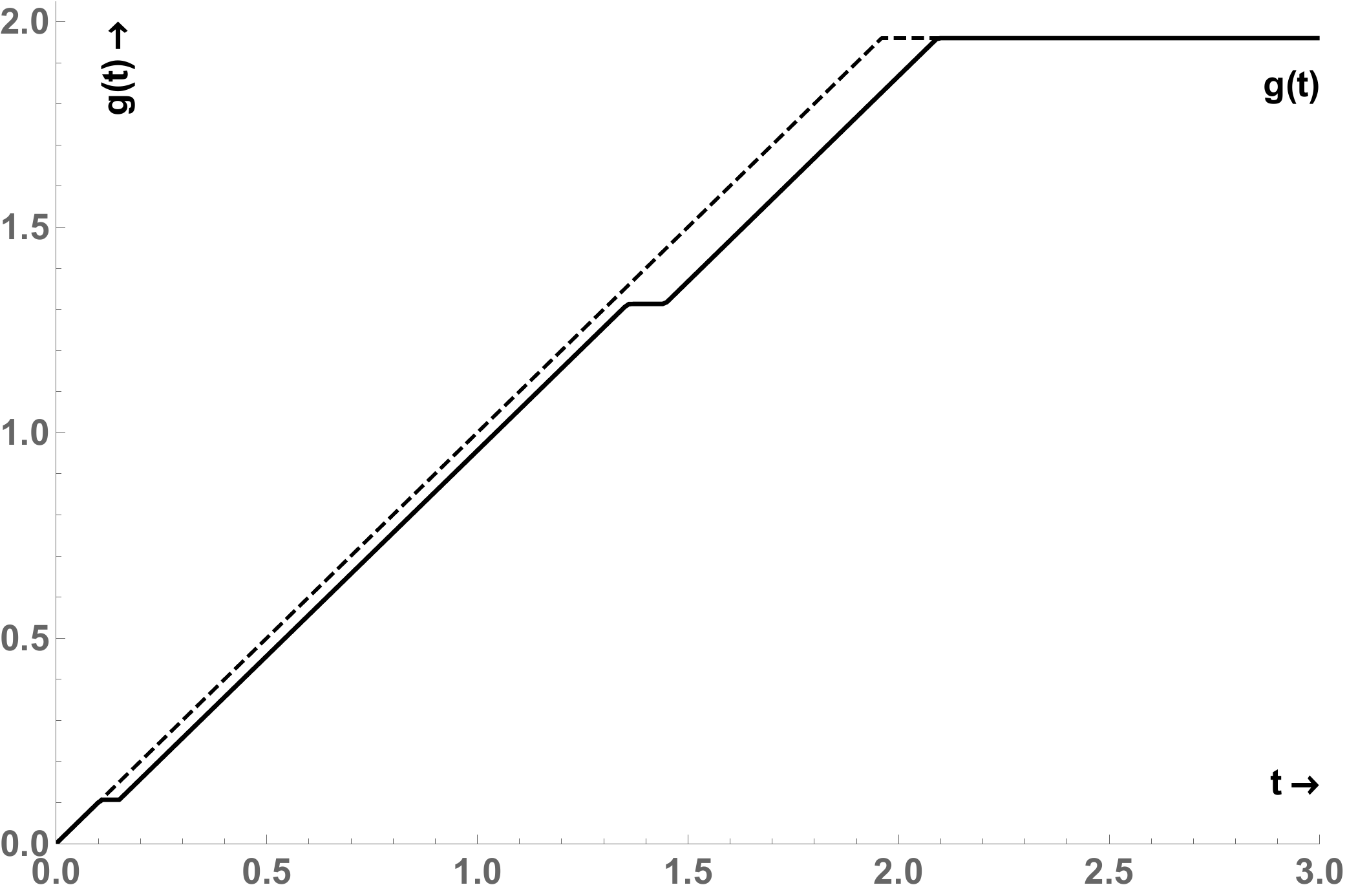} 
\vspace*{-4mm}
\caption[shot]{Optimal $g$-boundary function. The dashed line is the LR
boundary.}
\label{fig:OptimalVG}
\end{figure}
\begin{figure}[h]
\par
\begin{center}
\includegraphics[width=0.7\textwidth]{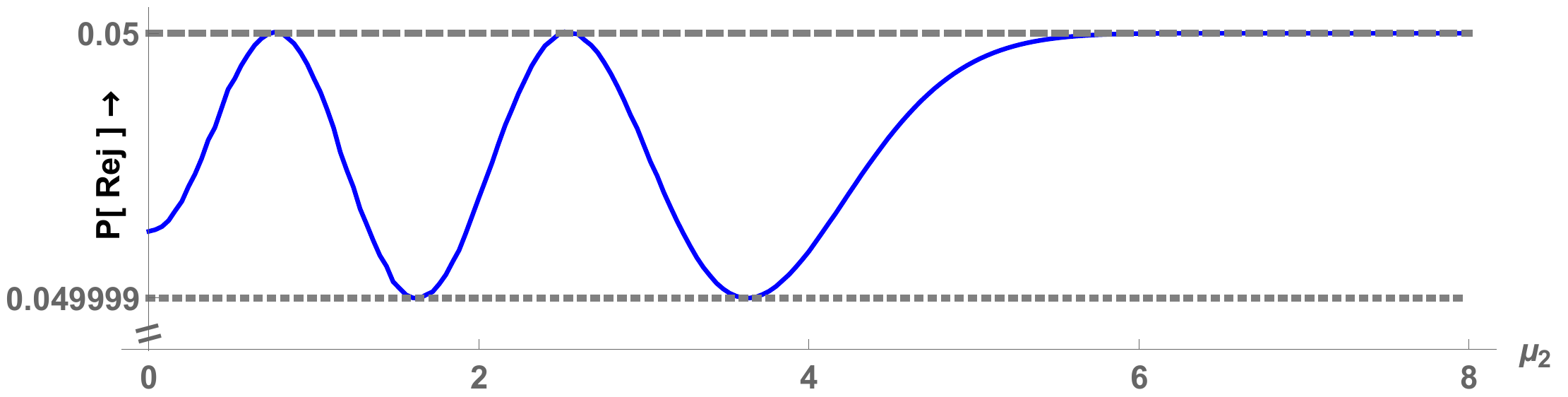}
\end{center}
\par
\vspace*{-6mm}
\caption[shot]{NRP as a function of the noncentrality parameter $\protect\mu 
$. The solid line is the optimal test and is strictly between 0.049999 ($%
=0.05-\protect\epsilon ,$ dotted line) and 0.05 (dashed line).}
\label{fig:NRPoptVGmagnified}
\end{figure}

The optimal $CR_{g}$ includes a narrow region close to the $45^{\circ }$
line where both $t$-statistics are of the same magnitude. This is expedient
for two reasons. First, because mediation requires both $\theta _{1}$ and $%
\theta _{2}$ to be non-zero. The best possibility of detecting this is along
the $45^{\circ }$ line as illustrated by the power surface in Figure \ref%
{fig:PowerSurfVGoptimalsolution} showing highest power on the diagonal. The
optimal $CR_{g}$ does exclude both $t$-statistics smaller than $0.1$, unlike
the unappealing region of the exact test. Second, the near similarity
condition requires additional critical region area in the left corner of the
octant because NRPs are particularly low for small parameter values. The
increased power is naturally linked to the increase in Type I error, but
correct size of a test by definition merely requires that this is not larger
than $5\%$. Nevertheless, size (NRP)/power trade-off exists as well as other
compromises that can be assessed using critical region analysis. For
instance, it may seem less intuitive that rejection is not monotonic in $%
t_{1}$ and $t_{2}$ since an increase in both $t_{1}$ and $t_{2}$ represents
increased evidence against the null. The LR and Wald tests are monotonic in
this sense, but lead to a reduction in power to nearly zero for small
parameter values. No observed value $t$ of $T$ will ever lie on the
horizontal or vertical axis. Any observed $t$ is therefore more likely given
an alternative parameter value than a value under the null. It is therefore
desirable to add area to the LR\ critical region even if this results in a
non-convex critical region or acceptance region. One could cogitate about
the very narrow region close to the diagonal and whether the acceptance
should not continue along the $45^{\circ }$ line further than $0.1,$ until
e.g. $1.217$ as in Figure \ref{fig:CRLRs}, but the new $g$-boundary is the
optimal solution to a well-defined problem.

The narrow region of the optimal $CR_{g}$ is a strict extension of the $%
CR_{LR},$ which itself is strictly larger than the Sobel (Wald) $CR_{W}$.
Since the new test is constructed to satisfy the size condition we have the
following: 

\begin{theorem}
The Sobel/Wald test and the LR test are inadmissible.
\end{theorem}

\begin{proof}
$CR_{W}\subset CR_{LR}\subset CR_{g}$ hence $P[CR_{W}]<P[CR_{LR}]<P[CR_{g}]%
\leq 0.05$. The optimal $g$-test has uniformly higher power and is correctly
sized by construction.
\end{proof}

The NRP as a function of the noncentrality parameter $\mu $ is shown in
Figure \ref{fig:NRPoptVGmagnified}. The difference from $5\%$ is less than $%
10^{-5}$ and so small that the scale had to be magnified greatly, to an
extend that prevents comparison with the LR and Sobel tests in the same
graph.

The power of the new $g$-test is very close to the power envelope (upper
bound). The maximal difference is $0.00614$. This has important
implications. First, the upper bound is tight as claimed earlier. Second,
the upper bound of the power envelope and power surface of the $g$-test look
almost identical when graphed. Figure \ref{fig:PowerSurfVGoptimalsolution}
therefore shows only the power surface of the optimal $g$-test. Finally, the
new $g$-test is optimal for all intents and purposes in a larger class of
tests. It is optimal by construction within the class of near similar tests $%
\Gamma _{\alpha ,\epsilon }^{\mathbb{M}_{0}}$, but given the closeness of
its power surface to the (non)similar power envelope, there cannot exist any
near similar test that has additional power more than $0.00614$, even if
construction is based on a different method. More generally, nonsimilar
tests can have $2\%$ points more power, but at the possible cost of odd
rejection regions and low power for other parameter values that are not used
in the construction of the test. The optimal $g$-test has good properties
for \emph{all} parameter values.

The power surface in Figure \ref{fig:PowerSurfVGoptimalsolution} shows only
the first quadrant of the parameter space of $\left( \mu _{1},\mu
_{2}\right) .$ The other three quadrants follow by simple permutations and
reflections of the parameters.

\begin{figure}[h]
\par
\begin{center}
\includegraphics[width=0.6\textwidth]{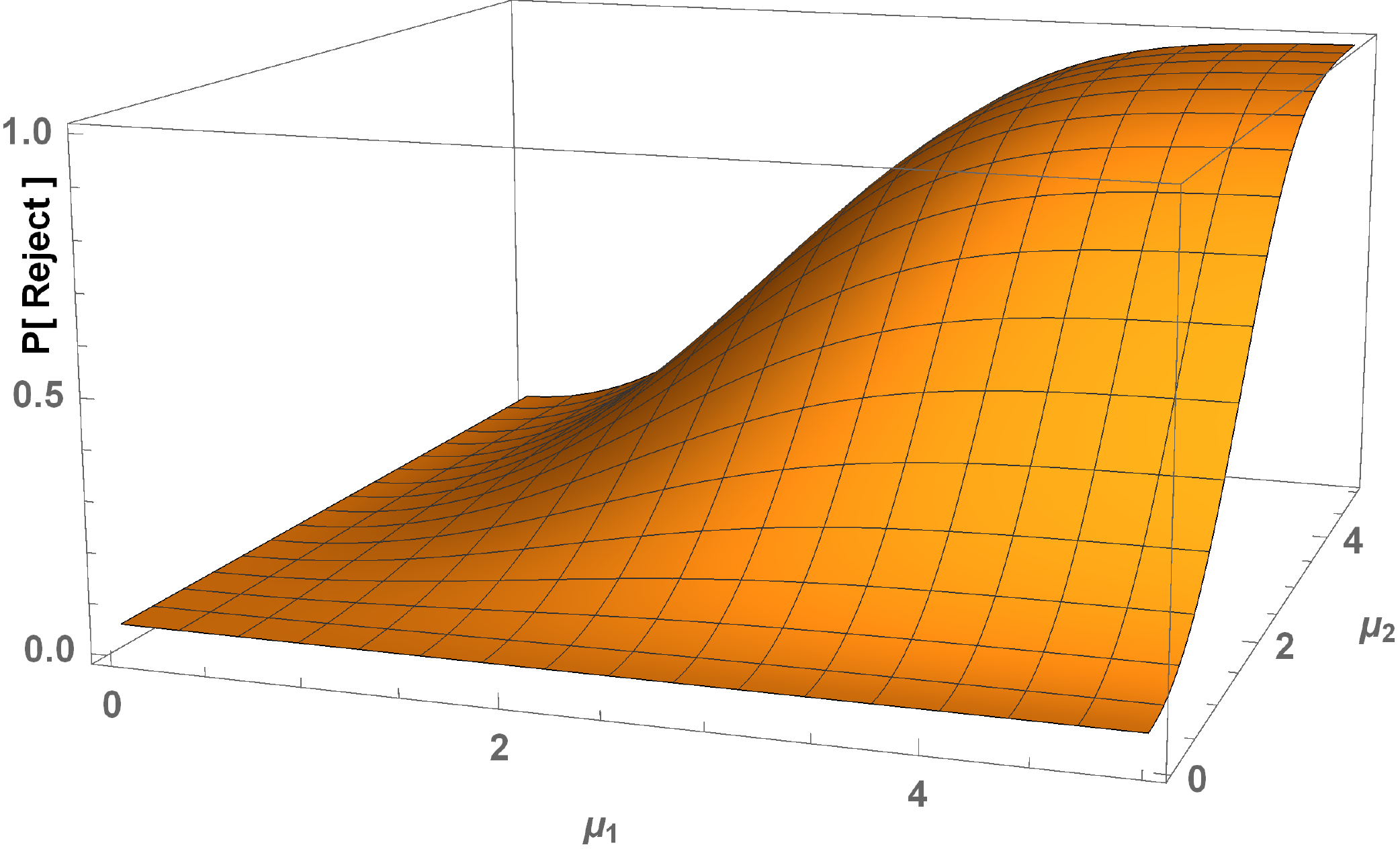}
\end{center}
\par
\vspace*{-6mm} 
\caption[shot]{Power surface optimal $g$-test }
\label{fig:PowerSurfVGoptimalsolution}
\end{figure}

\section{Application: Educational Attainment and Gender}

\label{sec:TeachApplication}

To illustrate the usage of the new test, we consider the mediation analysis
in \citet{Alan2018} on the effect of elementary school teachers' gender
beliefs, either traditional or progressive, on student mathematical and
verbal achievements. They exploit the unique institutional features of
Turkish data that provides a natural experiment in the random allocation of
teachers to schools. Information consists of approximately 4,000 third- and
fourth-grade students and their 145 teachers. The data are available in
their online appendix. Children are divided into three groups depending on
the length of exposure to a participating teacher: \textquotedblleft 1-year
exposure\textquotedblright\ (at most one year), \textquotedblleft 2-3 year
exposure\textquotedblright\ (more than one year and at most three years) and
\textquotedblleft 4-year exposure\textquotedblright\ (at most four years).\ %
\citet{Alan2018} consider three potential mediators, but we focus on
students' own gender role beliefs. In the notation of equations (1)-(3), $X$
is a dummy whether the teacher is classified as traditional or progressive.
The mediating variable $M$ is the student's own belief on gender roles. We
focus on the verbal test scores as the dependent variable $Y$. Table \ref%
{table:Alan} shows the estimates for the indirect effect for girls based on
the full sample and the three exposure groups after controlling for school
fixed effects, student characteristics, family characteristics, teacher
characteristics, teacher styles and teacher effort.

The results for the full sample are similar to the values shown in Table 5
and Table 6 of \citet{Alan2018}, although they use the approach by %
\citet{imai2013experimental}. We consider three different exposure duration
groups. The $t$-ratios of $\hat{\theta}_{1}$ are not significant at the $5\%$
level for more than 1-year exposure. For the 1-year exposure, however, it is
significant, but the $t$-ratio $t_{2}=-1.941$ of $\hat{\theta}_{2}$ is not.
So, the LR test would not find a significant effect and neither does the
simulation based method \citet{Alan2018} use. Using the new test, however,
we have $\left\vert t\right\vert _{(1)}=1.941$ and $|t|_{(2)}=2.052$, such
that $g(2.05)=1.9175$ and consequently $|t|_{(1)}>g(|t|_{(2)})$ and the
proposed test finds the mediation effect significant at the $5\%$ level. The
R function in Appendix E can be used if greater precision is desired, e.g. $%
g(2.052)=1.9195$.

\begin{table}[h]
\begin{center}
\begin{tabular}{llllllll}
Exposure: & Full &  & 1-Year &  & 2-3 Year &  & 4-Year \\ 
\cline{2-2}\cline{4-4}\cline{6-6}\cline{8-8}
$\hat{\theta _{1}}:$ & 0.199 &  & 0.256 &  & 0.109 &  & 0.064 \\ 
$t$-ratio: & \textbf{3.140} &  & \textbf{2.052} &  & \textbf{1.065} &  & 
\textbf{0.513} \\ 
$\hat{\theta _{2}}:$ & -0.119 &  & -0.097 &  & -0.125 &  & -0.113 \\ 
$t$-ratio: & \textbf{-5.343} &  & \textbf{-1.941} &  & \textbf{-4.163} &  & 
\textbf{-1.931} \\ 
$\hat{\theta _{1}}\cdot \hat{\theta _{1}}$: & -0.024 &  & -0.025 &  & -0.014
&  & -0.007%
\end{tabular}%
\end{center}
\caption{Mediation effect of students' own gender role beliefs on verbal
test scores by exposure to progressive or traditional teachers. Full-sample
results are similar to the values of Table 5 (Gender Role Beliefs: Girls)
and Table 6 (Panel B. Verbal Test Scores: Gender Role Beliefs) of 
\citet{Alan2018} who use the approach by \citet{Imaietal2010}, which is
slightly different, but the value of the indirect effect in the last row is
numerically the same as ours based on the regressions (\protect\ref%
{eq:unrstricted_Y}) and (\protect\ref{eq:unrstricted_I}) extended with the
same controls.}
\label{table:Alan}
\end{table}

\section{Higher Dimensions}

\label{sec:HigherDimensions}

In a further empirical illustration below, mediation may be via channels
that involve two mediating variables. This requires a multivariate extension
of the new test. The necessary invariance and distribution theory was given
in Section \ref{sec:Theory2} but a further aspect is the coherency between
solutions in different dimensions. Consider the null hypothesis $%
H_{0}:\theta _{1}\theta _{2}\cdots \theta _{K}=0$ in $K$ dimensions. If it
were known that $\theta _{K}\neq 0,$ then the null hypothesis reduces to $%
H_{0}:\theta _{1}\theta _{2}\cdots \theta _{K-1}=0.$ This implies that the
critical region for the $K-1$ corresponding $t$-statistics should reduce to
the solution found for $K-1$ dimensions when $\left\vert \mu _{K}\right\vert 
$ is large. For large values of $\left\vert T_{K}\right\vert $ (very small $%
p $-values) it is essentially known that $\mu _{K}$ and $\theta _{K}$ are
non-zero. The probability of rejection will effectively depend only on the $%
K-1$ other $t$-values. In two dimensions this means that as $t\rightarrow
\infty $ the boundary function $g\left( t\right) \rightarrow 1.96$, 
which is the one-dimensional solution for testing $H_{0}:\theta _{1}=0$ when 
$\alpha =5\%.$ In three dimensions it means that the solution must reduce to
the $g$-test derived in Section \ref{sec:OptimalgTest}. We make this
requirement explicit in the definition below.

\begin{definition}
The $g$-boundary in dimension $K$ for the test that rejects if $\newline
\left\vert T\right\vert _{\left( 1\right) }>g\left( \left\vert T\right\vert
_{\left( 2\right) },\cdots \left\vert T\right\vert _{\left( K\right)
}\right) $ is dimensionally coherent if for each $2\leq k\leq K$ 
\begin{equation*}
\lim_{t_{k}\rightarrow \infty }g\left( t_{2},\cdots ,t_{k-1},t_{k}\right)
=g\left( t_{2},\cdots ,t_{k-1}\right)
\end{equation*}
\end{definition}

We used a multivariate spline generalization to implement the varying $g$%
-method based on barycentric coordinates in three dimensions. It resulted in
a maximum of $0.2\%$ points difference from $5\%$. But imposing dimensional
coherency was complicated and the problem suffers from the curse of
dimensionality: the dimension of the integral increases with $K$ and the
number of knots necessary to define $g$ impedes optimization. For practical
purposes we therefore propose a simple solution that exploits the
dimensional coherency inductively and weighs the LR test in dimension $K$
with the solution obtained in dimension $K-1$. First note that one could
satisfy the coherency condition in three dimensions by simply rejecting when 
$\left\vert T\right\vert _{\left( 1\right) }>g(\left\vert T\right\vert
_{\left( 2\right) }),$ irrespective of $\left\vert T\right\vert _{\left(
3\right) }.$ This results in an invalid test however, because it is
oversized with a maximum NRP of $7.2\%$ when $K=3$. The LR test on the other
hand is conservative, in particular near the origin. A practical solution
therefore is to use a weighted average between the liberal and conservative
boundary. We use weights that depend on the largest absolute $t$-statistic.
In particular for $K=3$ 
\begin{equation*}
g(t_{2},t_{3})=(1-w(t_{3}))g_{LR}(t_{2})+w(t_{3})g(t_{2}),\ \ 
\end{equation*}%
with weight function $w(t_{3})$ a linear spline with $0\leq w(t_{3})\leq 1,$
and $\lim_{t_{3}\rightarrow \infty }w(t_{3})=1$. The test rejects if $%
\left\vert T\right\vert _{\left( 1\right) }>g(\left\vert T\right\vert
_{\left( 2\right) },\left\vert T\right\vert _{\left( 3\right) }).$
Minimizing deviation ofthe NRP from the significance level and imposing the
size restriction, results in a spline with knots: 
\begin{equation*}
\{(0,0)\},\{(1.35,0.959),(2.025,0.842),(2.7,1),(\infty ,1)\},
\end{equation*}%
leading to a maximum of $0.13\%$ points difference in NRPs from $5\%$ and
never exceeding $5\%$. The solution is shown in Figure \ref{fig:3Dsolution}.

\begin{figure}[h]
\par
\begin{center}
\includegraphics[width=3.in, height=3.in, trim = 0mm 0mm 0mm
0mm,clip]{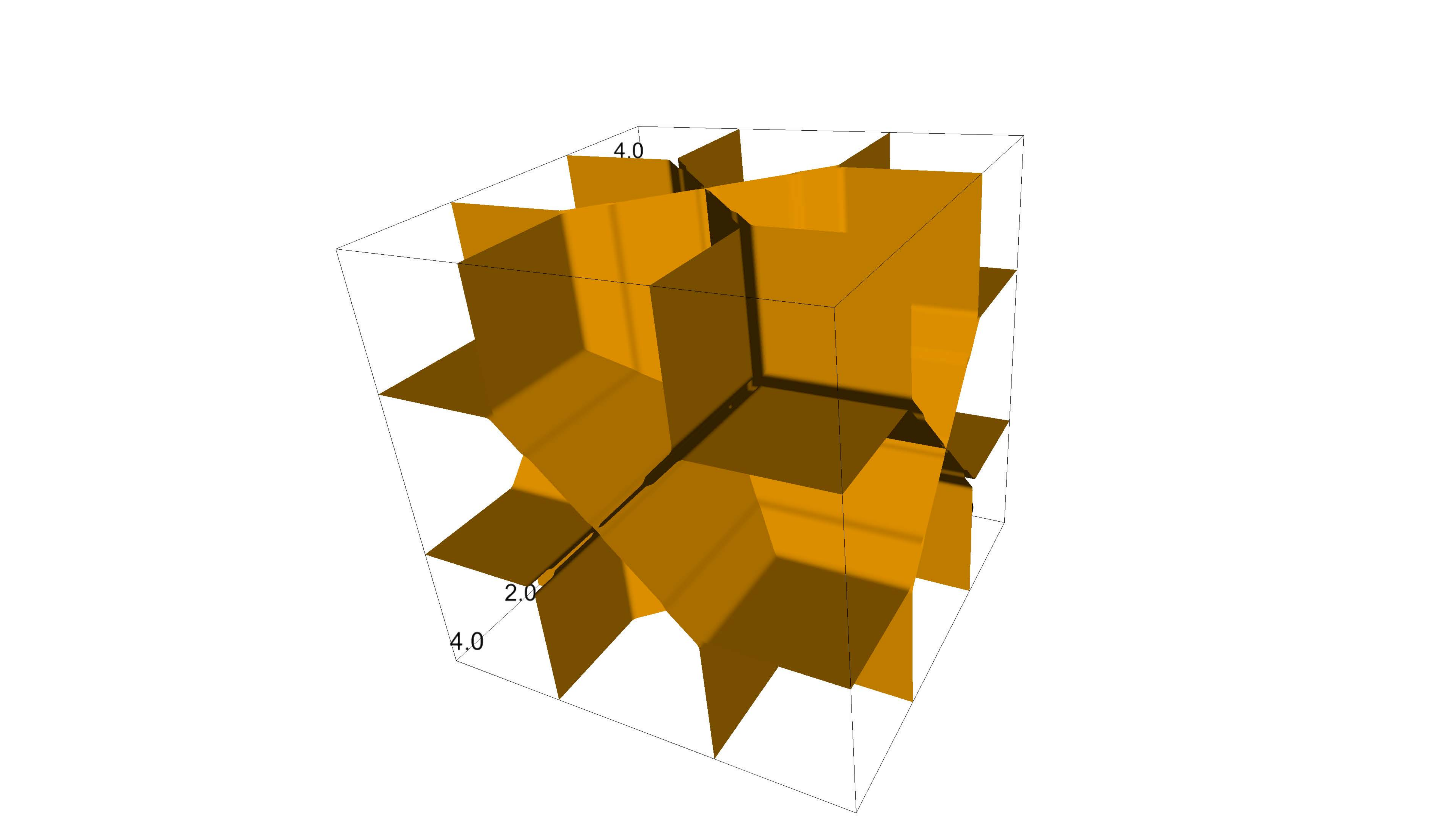}
\end{center}
\par
\vspace*{-8mm}
\caption{The $g$-boundary in 3D for the unsorted $t$-statistics $\left(
T_{1},T_{2},T_{3}\right) $. CR is furthest removed from the origin and
includes e.g. (4,4,4). The edges show the 2D solution since one $t$%
-statistic is very large. If two $t$-statistics are very large then it
reduces to 1.96, the 1D solution.}
\label{fig:3Dsolution}
\end{figure}

\section{Empirical Illustration}

\label{sec:EmpiricalIllustration}

For a numerical illustration, we consider the recursive model of union
sentiment among southern nonunion textile workers as used by \cite%
{Bollen1990}. The model:

\begin{equation}
\left[ 
\begin{array}{l}
y \\ 
m_{2} \\ 
m_{1}%
\end{array}%
\right] =\left[ 
\begin{array}{lll}
0 & \beta _{12} & \beta _{13} \\ 
0 & 0 & \beta _{23} \\ 
0 & 0 & 0%
\end{array}%
\right] \left[ 
\begin{array}{l}
y \\ 
m_{2} \\ 
m_{1}%
\end{array}%
\right] +\left[ 
\begin{array}{ll}
\tau _{11} & 0 \\ 
0 & \alpha _{22} \\ 
0 & \alpha _{32}%
\end{array}%
\right] \left[ 
\begin{array}{l}
x_{1} \\ 
x_{2}%
\end{array}%
\right] +\left[ 
\begin{array}{l}
u \\ 
v_{1} \\ 
v_{2}%
\end{array}%
\right] ,  \label{eq:unionsentimentstructural}
\end{equation}%
is a simplified version of \citet{McDonald1984} and discussed in some detail
by \citet[p. 82--93]{Bollen1989}. It analyses the direct and indirect
effects of tenure and age on union sentiment via deference and/or labor
activism. Tenure $x_{1}$ is measured in log of years working in a particular
textile mill and age $x_{2}$ is measured in years. The variables sentiment
towards unions $y$, deference/submissiveness to managers $m_{1}$, and
support for labor activism $m_{2}$, are measures based on 7, 4, and 9 survey
questions respectively. The disturbances ($u$, $v_{1}$ and $v_{2}$) are
assumed to be uncorrelated across equations and individuals. When they are
normally distributed, ML estimation of the system reduces to OLS applied to
each equation separately due to the recursive structure.

\begin{figure}[hb]
\par
\begin{center}
\includegraphics[width=0.7\textwidth]{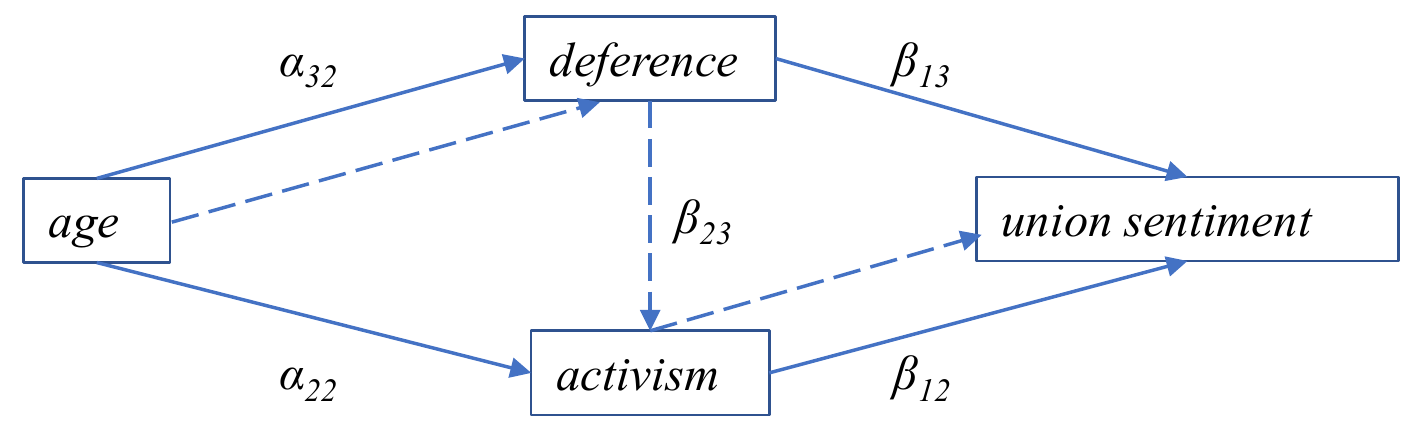}
\end{center}
\par
\vspace*{-6mm}
\caption{Union sentiment mediation graph}
\label{fig:PathDiagram4txt}
\end{figure}

We use a selection of 100 observations out of the original 173 and focus on
three alternative theories of the indirect effects from age to union
sentiment: two competing parallel effects that the age effect is mediated by
increased deference in which case $i_{1}=\alpha _{32}\beta _{13}$ quantifies
the indirect effect. The alternative mediation channel is that activism
mediates such that $i_{2}=\alpha _{22}\beta _{12}$ is the indirect effect.
The third channel is a serial effect that age affects deference, which in
turn affects activism, which in turn affects union sentiment such that $%
i_{3}=\alpha _{32}\beta _{12}\beta _{23}$ measures the indirect effect.
Figure~\ref{fig:PathDiagram4txt} illustrates the three mediation channels.
The OLS estimates of the coefficients of the structural equations and their $%
t$-statistics are shown in Table \ref{tab:UnionModelOLS}.

\begin{table}%

\centering%
\begin{tabular}{lllllll}
& $\alpha _{32}$ & $\alpha _{22}$ & $\beta _{23}$ & $\beta _{12}$ & $\beta
_{13}$ & $\tau _{11}$ \\ \hline
Estimate & --0.050 & 0.057 & --0.283 & 0.987 & --0.215 & 0.720 \\ 
$t$-statistic & --1.902 & 2.709 & --3.582 & 7.120 & --1.838 & 1.777 \\ \hline
&  &  &  &  &  & 
\end{tabular}
\vspace*{-6mm} 
\caption{OLS Estimates and $t$-statistics for Union Sentiment Model
($N=100$)} \label{tab:UnionModelOLS} 
\end{table}%

The point estimates of the indirect effects and their $t$-statistics based
on the delta method are shown in Table \ref{tab:UnionIndirecteffecttests}.
For the $g$-test we need the absolute order statistics, evaluate $g,$ and
compare. For $H_{0}^{i_{1}}:\alpha _{32}\beta _{13}=0$\ we observe $|t(\hat{%
\beta}_{13})|=1.838>1.770=g(1.902)=g(|t(\hat{\alpha}_{32})|),$ and hence
reject. For $H_{0}^{i_{2}}:\alpha _{22}\beta _{12}=0$ \ we have $|t(\hat{%
\alpha}_{22})|=2.709>1.960=g(7.120)=g(|t\left( \hat{\beta}_{12}\right) |)$
and also reject. Testing the last null hypothesis $H_{0}^{i_{3}}:\alpha
_{32}\beta _{12}\beta _{23}=0$ requires the three-dimensional solution given
in Figure \ref{fig:3Dsolution}. We have $|t(\hat{\alpha}%
_{32})|=1.902<1.96=g_{2}(3.582,7.120)=g_{2}(|t\left( \hat{\beta}%
_{23}|\right) ,|t\left( \hat{\beta}_{12}\right) |)$ and do not reject.

\begin{table}%
\centering%
\begin{tabular}{llll}
& Estimate & Sobel $t$-statistic & $g$-test \\ \hline
$i_{1}=\alpha _{32}\beta _{13}$ & $0.0108$ & $1.322$ & $1.838^{\ast
}>1.770=g(|-1.902|)$ \\ 
$i_{2}=\alpha _{22}\beta _{12}$ & $0.0561$ & $2.532^{\ast }$ & $2.709^{\ast
}>1.960=g(7.120)$ \\ 
$i_{3}=\alpha _{32}\beta _{12}\beta _{23}$ & $0.0140$ & $1.635$ & $%
1.902\ngtr 1.960=g_{2}(3.582,7.120)$ \\ \hline
&  &  & 
\end{tabular}%
\vspace*{-6mm} 
\caption{Estimates, Sobel $t$-statistics and $g$-test. *
indicates significance at $5\%$} \label{tab:UnionIndirecteffecttests} 
\end{table}%
\smallskip

The Sobel test with critical value $1.96$ concludes that $i_{2}$ is
significant but does not find enough evidence for the $i_{1}$ mediation
channel. The new $g$-test in contrast, concludes that $i_{1}$ is also
significant. Both $t$-values in this case are smaller than $1.96$, so the LR
test would not reject either. The two $t$-values are of comparable magnitude
and the $g$-test finds a significant mediation effect. For implementation of
the $g$-test only the relevant $t$-statistics are required. The absolute
values are ordered and the smallest value compared with the value of the $g$%
-function evaluated at the largest absolute $t$-value. This can be looked up
in Table~\ref{table:g_values_intro} (possibly using linear interpolation) or
one can use the R code provided in Appendix \ref{sec:AppendixRCode}.%
\footnote{%
The bootstrap is a popular alternative for testing mediation. Because of the
asymmetry of the distribution involved this is carried out through
alternative confidence intervals of the indirect effect. See e.g. %
\citet{MacKinnon2004} and \citet{Preacher2008}. The bootstrap is not valid.
Simulations we carried out showed that bootstrap tests for mediation based
on generally preferred BCa confidence intervals can have sizes of 8\% when $%
n=100$ and higher for $n$ smaller.}

For $i_{3}$ both tests draw the same conclusion. The three $t$-values
involved are not of comparable magnitude and the $t$-statistics for $\beta
_{12}$ and $\beta _{23}$ are so large that rejecting the null $H_{0}:\alpha
_{32}\beta _{12}\beta _{23}=0$ essentially depends on whether $\alpha _{32}$
is zero.\ The corresponding absolute $t$-value of $1.90$ is too small to
warrant such conclusion.

\section{Conclusion}

This paper proposes a new near similar more powerful mediation test that is
simple to used based on two ordinary $t$-statistics.\ The mediation problem
is empirically extremely important in many different fields. Theoretically
we solve an interesting statistical problem which has generated results
dating back to \citet{Craig1936} and still continues today with
contributions on poor performance of the Wald statistic, construction of
similar tests, and hypotheses with singularities. A main theoretical
contribution has been the derivation of an exact similar test that is unique
within the general class considered (with CR that is topologically simply
connected and has weak monotone c\`{a}dl\'{a}g boundary). This exact test
has unattractive statistical properties, however, leading us to consider a
class of nearly similar tests and choose an attractive test within it. By
relaxing the strict similarity condition we are able to construct a near
similar test that has superior power and avoids the disagreeable properties
of the exact similar test and would please even statistically erudite
emperors in \citet{Perlman1999}. The new test can also be justified
asymptotically under much weaker conditions and other estimation methods.

The new test is constructed using a new general method we propose for
constructing tests that are approximately similar. This varying-$g$ method
considers a flexible critical region boundary and minimizes the difference
from the level $\alpha $ of the rejection probabilities at a number of
points on the boundary of the null hypothesis. Conceptually and practically
this was very simple and straightforward to implement. It does not require,
as in other approaches, a choice of mixture distribution, nor the
construction of least favorable distributions. Numerically it is also
attractive in terms of convergence properties and avoids the need for
simulations. The appropriate distribution theory for the mediation case and
higher dimensional extensions is explicitly given. All our calculations are
done using numerical integration using the distribution of the maximal
invariant.

The new method is applicable to many other testing problems with nuisance
parameters. It is remarkable that this simple method works so well and can
deliver substantial improvements. Even the simplest linear interpolation
implementation for the mediation hypothesis increases power by almost $5\%$
points when mediation effects are small.

We have calculated a power envelope upper bound for the mediation testing
problem that is very tight. Using this result, we were able to construct a
test that is optimal within the class of near similar tests. It minimizes
the total difference between its power surface and the power envelope bound.
It results in a point wise difference less than $0.0062$ for all alternative
parameter points considered. This implies that the test is practically
optimal even if a more general class of possible test construction is
considered. A power envelope for nonsimilar tests showed that power loss due
to the similarity requirement is minimal since the maximum power loss is
less than $2\%$ points (when maximum power is around $40\%$).

The optimal $g$-test satisfies the size condition. The critical region is
strictly larger than the LR and Wald critical regions and is therefore
strictly and uniformly more powerful. The classic tests are therefore not
admissible and their bootstrapped variants are not valid. For large values
of the standardized coefficients the power difference becomes negligible,
but when mediation effects are small or have relatively large standard
errors, the power can be close to $5\%$ points higher than these classic $%
5\% $-level tests. This has important consequences for empirical work. It
enables researcher to prove mediation effects earlier in circumstances that
one could not show mediation before due to extreme conservativeness of
standard tests near the origin.

\clearpage\pagebreak

\begin{appendices}%

\section{Theory}

\label{sec:AppendixTheory}

\subsection{Elementary Relation}

Let $y=\left( y_{1},\cdots ,y_{n}\right) ^{\prime }$, $m=\left( m_{1},\cdots
,m_{n}\right) ^{\prime }$, $x=\left( x_{1},\cdots ,x_{n}\right) ^{\prime }$,
be vectors of observables in deviations from their means such that $\bar{y}%
=0 $, $\bar{m}=0$, $\bar{x}=0$ and disturbance vectors $u=\left(
u_{1},\cdots ,u_{n}\right) ^{\prime }$, $v=\left( v_{1},\cdots ,v_{n}\right)
^{\prime }.$ The model is then: 
\begin{align}
y_{i}& =\tau x_{i}+\theta _{2}m_{i}+u_{i},  \label{eq:A:unrestr1Y} \\
m_{i}& =\theta _{1}x_{i}+v_{i},  \label{eq:A:unrestr2M}
\end{align}%
and the restricted version of Equation (\ref{eq:A:unrestr1Y}) with $\theta
_{2}=0$ equals:%
\begin{equation}
y_{i}=\tau ^{\ast }x_{i}+w_{i}.
\end{equation}

The claim $\hat{\tau}^{\ast }=\hat{\tau}+\hat{\theta}_{1}\hat{\theta}_{2}$
follows from a standard exercise to relate restricted and unrestricted OLS
estimators: 
\begin{equation*}
\widehat{\tau ^{\ast }}=\left( x^{\prime }x\right) ^{-1}x^{\prime }y=\left(
x^{\prime }x\right) ^{-1}x^{\prime }\left( x\hat{\tau}+m\hat{\theta}_{2}+%
\hat{u}\right) =\hat{\tau}+\left( x^{\prime }x\right) ^{-1}x^{\prime }m\hat{%
\theta}_{2}+\left( x^{\prime }x\right) ^{-1}x^{\prime }\hat{u},
\end{equation*}%
and $\hat{\theta}_{1}=\left( x^{\prime }x\right) ^{-1}x^{\prime }m$ is the
OLS estimator in Equation (\ref{eq:A:unrestr2M}) and $x^{\prime }\hat{u}=0$
since $\hat{u}$ are the OLS residuals from Equation (\ref{eq:A:unrestr1Y})
and orthogonal to $x.$ \newline
The parameter relation $\tau ^{\ast }-\tau =\theta _{1}\theta _{2}$ follows
by substituting (\ref{eq:A:unrestr2M}) in (\ref{eq:A:unrestr1Y}):%
\begin{equation*}
y_{i}=\tau x_{i}+\theta _{2}m_{i}+u_{i}=\left( \tau +\theta _{2}\theta
_{1}\right) x_{i}+\left( \theta _{2}v_{i}+u_{i}\right) =\tau ^{\ast
}x_{i}+w_{i}.
\end{equation*}%
It follows that $H_{0}:\theta _{1}\theta _{2}=0\Leftrightarrow H_{0}:\tau
^{\ast }=\tau .$

\subsection{Likelihood}

The joint density of $\left( y,m\right) $ given $x$ is $f\left( \left.
y,m\right\vert \ x\right) =f\left( \left. y\right\vert \ m,x\right) f\left(
m,x\right) ,$ and according to the model:%
\begin{align*}
\left. y_{i}\right\vert m_{i},x_{i}& \sim N\left( \tau x_{i}+\theta
_{2}m_{i},\sigma _{11}\right) , \\
\left. m_{i}\right\vert x_{i}& \sim N\left( \theta _{1}x_{i},\sigma
_{22}\right) .
\end{align*}%
Hence the log-likelihood: $\ell =\ell \left( \tau ,\theta _{2},\sigma
_{11},\theta _{1},\sigma _{22}\right) =\log f\left( y|m,x;\tau ,\theta
_{2},\sigma _{11}\right) +\log f\left( m|x;\theta _{1},\sigma _{22}\right) ,$
for $n$ independent observations equals Equation (\ref{eq:loglik}) which can
be written as:%
\begin{align*}
\ell & \propto -\frac{1}{2\sigma _{11}}y^{\prime }y+\frac{\tau }{\sigma _{11}%
}y^{\prime }x+\frac{\theta _{2}}{\sigma _{11}}y^{\prime }m-\left( \frac{\tau
\theta _{2}}{\sigma _{11}}-\frac{\theta _{1}}{\sigma _{22}}\right) x^{\prime
}m-\frac{1}{2}\left( \frac{\theta _{2}^{2}}{\sigma _{11}}+\frac{1}{\sigma
_{22}}\right) m^{\prime }m+ \\
& ~~~\ \ \ ~~~\ \ \ ~~~\ \ \ ~~~\ \ \ ~~~\ \ \ ~~~\ \ \ ~~~\ \ \ ~~~\ \ \
~~~\ \ \ ~~~\ \ \ -\left( \frac{\tau ^{2}}{2\sigma _{11}}-\frac{\theta
_{1}^{2}}{2\sigma _{22}}\right) x^{\prime }x-\frac{n}{2}\log \left( \sigma
_{11}\sigma _{22}\right) \\
& =\eta ^{\prime }r-\kappa \left( \eta ,x^{\prime }x\right) ,~~~\ \ \ with:
\\
\eta & =\left( -\frac{1}{2\sigma _{11}},\frac{\tau }{\sigma _{11}},\frac{%
\theta _{2}}{\sigma _{11}},-\left( \frac{\tau \theta _{2}}{\sigma _{11}}-%
\frac{\theta _{1}}{\sigma _{22}}\right) ,-\frac{1}{2}\left( \frac{\theta
_{2}^{2}}{\sigma _{11}}+\frac{1}{\sigma _{22}}\right) \right) ^{\prime }, \\
r& =\left( y^{\prime }y,y^{\prime }x,y^{\prime }m,x^{\prime }m,m^{\prime
}m\right) ^{\prime },
\end{align*}%
and $\kappa $ some function of $\eta $ and $x^{\prime }x$ which is fixed.
Since $\dim \left( \eta \right) =\dim \left( r\right) $ the model is a full
exponential model of dimension five following the Koopman-Fisher-Darmois
theorem (see \citet{vanGarderen1997}) and $r$ is a complete sufficient
statistic. The score $\mathbf{s}\left( \tau ,\theta _{2},\sigma _{11},\theta
_{1},\sigma _{22}\right) =\mathbf{s}=\left( \mathbf{s}_{1}^{\prime },\mathbf{%
s}_{2}^{\prime }\right) ^{\prime }$ is analogues to the scores of the two
separate regression models since $\left( \tau ,\theta _{2},\sigma
_{11}\right) $ appears in the first equation only, and $\left( \theta
_{1},\sigma _{22}\right) $ appears in the second equation only. So:

$\mathbf{s=}\left( \frac{\left( y-\tau x-\theta _{2}m\right) ^{\prime }x}{%
\sigma _{11}},\frac{\left( y-\tau x-\theta _{2}m\right) ^{\prime }m}{\sigma
_{11}},\frac{\left( y-\tau x-\theta _{2}m\right) ^{\prime }\left( y-\tau
x-\theta _{2}m\right) }{2\sigma _{11}^{2}}-\frac{n}{2\sigma _{11}},\frac{%
\left( m-\theta _{1}x\right) ^{\prime }x}{\sigma _{22}},\frac{\left(
m-\theta _{1}x\right) ^{\prime }\left( m-\theta _{1}x\right) }{2\sigma
_{22}^{2}}-\frac{n}{2\sigma _{22}}\right) ^{\prime }$ \newline
and the Maximum Likelihood Estimator (MLE) equals the MLE for the two
equations separately:%
\begin{align*}
\left( 
\begin{array}{c}
\hat{\tau} \\ 
\hat{\theta}_{2}%
\end{array}%
\right) & =\left( \left( x:m\right) ^{\prime }\left( x:m\right) \right)
^{-1}\left( x:m\right) ^{\prime }y;~~~\ \ \ \widehat{\sigma }_{11}=\frac{1}{n%
}y^{\prime }M_{X}y; \\
\hat{\theta}_{1}& =\left( x^{\prime }x\right) ^{-1}x^{\prime }m;~~~\ \ \ 
\widehat{\sigma }_{22}=\frac{1}{n}m^{\prime }M_{x}m,
\end{align*}%
with $M_{A}=I-A\left( A^{\prime }A\right) ^{-1}A^{\prime }$ and $X=\left[ x:m%
\right] $ an $n\times 2$ matrix. The MLE is minimal sufficient and complete
because it is a bijective transformation of $r$ which is a minimal
sufficient and complete statistic.

\subsection{Classic Tests}

\textbf{Wald test.} Under $H_{0}:\theta _{1}\theta _{2}=\mathbf{r}\left(
\theta _{1},\theta _{2}\right) =0$. Then $R\left( \theta _{1},\theta
_{2}\right) =\frac{\partial \mathbf{r}\left( \theta _{1},\theta _{2}\right) 
}{\partial \left( \theta _{1},\theta _{2}\right) ^{\prime }}=\left( \theta
_{2},\theta _{1}\right) ^{\prime }$ and evaluated at the (unrestricted) MLE
equals $R\left( \hat{\theta}_{1},\hat{\theta}_{2}\right) =\left( \hat{\theta}%
_{2},\hat{\theta}_{1}\right) ^{\prime }.$ The Wald test therefore becomes: 
\begin{align*}
W& =\hat{\theta}_{1}\hat{\theta}_{2}\left( \left( 
\begin{array}{c}
\hat{\theta}_{2} \\ 
\hat{\theta}_{1}%
\end{array}%
\right) ^{\prime }\left( 
\begin{array}{ll}
\sigma _{\hat{\theta}_{1}}^{2} & 0 \\ 
0 & \sigma _{\hat{\theta}_{2}}^{2}%
\end{array}%
\right) \left( 
\begin{array}{c}
\hat{\theta}_{2} \\ 
\hat{\theta}_{1}%
\end{array}%
\right) \right) ^{-1}\hat{\theta}_{1}\hat{\theta}_{2} \\
& =\frac{\hat{\theta}_{1}^{2}\hat{\theta}_{2}^{2}}{\hat{\theta}%
_{1}^{2}\sigma _{\hat{\theta}_{2}}^{2}+\hat{\theta}_{2}^{2}\sigma _{\hat{%
\theta}_{1}}^{2}}\cdot \frac{\left( \sigma _{\hat{\theta}_{1}}^{2}\sigma _{%
\hat{\theta}_{2}}^{2}\right) ^{-1}}{\left( \sigma _{\hat{\theta}%
_{1}}^{2}\sigma _{\hat{\theta}_{2}}^{2}\right) ^{-1}}=\frac{%
T_{1}^{2}T_{2}^{2}}{T_{1}^{2}+T_{2}^{2}}
\end{align*}%
The Sobel test equals $\sqrt{W}$ and is usually expressed as the square root
of the first term in the second line:$\frac{\hat{\theta}_{1}\hat{\theta}_{2}%
}{\sqrt{\hat{\theta}_{1}^{2}\sigma _{\hat{\theta}_{2}}^{2}+\hat{\theta}%
_{2}^{2}\sigma _{\hat{\theta}_{1}}^{2}}}$.

\textbf{LR test.} The maximum value of the log-likelihood can be expressed
in terms of the OLS\ residual sum of squares in the usual way for the first
and second equation, $RSS_{1}$ and $RSS_{2}$ respectively:%
\begin{align*}
\ell \left( \hat{\tau},\hat{\theta}_{2},\hat{\sigma}_{11},\hat{\theta}_{1},%
\hat{\sigma}_{22}\right) & \propto -\frac{1}{2\hat{\sigma}_{11}}%
\sum_{i=1}^{n}\left( y_{i}-\hat{\tau}x_{i}-\hat{\theta}_{2}m_{i}\right) ^{2}-%
\frac{n}{2}\log \left( \hat{\sigma}_{11}\right) + \\
& ~~~\ \ \ ~~~\ \ \ ~~~\ \ \ ~~~\ \ \ ~~~\ \ \ -\frac{1}{2\hat{\sigma}_{22}}%
\sum_{i=1}^{n}\left( m_{i}-\hat{\theta}_{1}x_{i}\right) ^{2}-\frac{n}{2}\log
\left( \hat{\sigma}_{22}\right) \\
& =-\frac{n}{2}-\frac{n}{2}\log \left( RSS_{1}/n\right) -\frac{n}{2}-\frac{n%
}{2}\log \left( RSS_{2}/n\right) .
\end{align*}%
Denote the restricted residual sums of squares by $\widetilde{RSS}_{1}$ when 
$\theta _{2}=0$, $\widetilde{RSS}_{2}$ when $\theta _{1}=0,$ and the
restricted maximized log-likelihoods by: 
\begin{align*}
\ell _{\theta _{1}=0}\left( \hat{\tau},\hat{\theta}_{2},\hat{\sigma}_{11},0,%
\tilde{\sigma}_{22}\right) & \propto -\frac{n}{2}-\frac{n}{2}\log \left(
RSS_{1}/n\right) -\frac{n}{2}-\frac{n}{2}\log \left( \widetilde{RSS}%
_{2}/n\right) , \\
\ell _{\theta _{2}=0}\left( \tilde{\tau},0,\tilde{\sigma}_{11},\hat{\theta}%
_{1},\hat{\sigma}_{22}\right) & \propto -\frac{n}{2}-\frac{n}{2}\log \left( 
\widetilde{RSS}_{1}/n\right) -\frac{n}{2}-\frac{n}{2}\log \left(
RSS_{2}/n\right) .
\end{align*}%
The LR\ test of the full model with five parameters, against the model with
the single restriction $\theta _{2}=0$ equals: 
\begin{align*}
LR_{\theta _{2}=0}& =2\left( -\frac{n}{2}\log \left( RSS_{1}/n\right) +\frac{%
n}{2}\log \left( \widetilde{RSS}_{1}/n\right) \right) =n\log \left( 1+\frac{1%
}{n}T_{2}^{2}\right) \text{ since} \\
\widetilde{RSS}_{1}& =\hat{\theta}_{2}^{\prime }m\left( m^{\prime
}M_{x}m\right) ^{-1}m\hat{\theta}_{2}+RSS_{1}\text{ \ and} \\
T_{2}^{2}& =\hat{\theta}_{2}^{\prime }m\left( m^{\prime }M_{x}m\right) ^{-1}m%
\hat{\theta}_{2}/\hat{\sigma}_{11}=\left( \widetilde{RSS}_{1}-RSS_{1}\right)
/\left( RSS_{1}/n\right) .
\end{align*}%
Analogously the LR test for $\theta _{1}=0$ equals:%
\begin{equation*}
LR_{\theta _{1}=0}=n\log \left( 1+\frac{1}{n}T_{1}^{2}\right) .
\end{equation*}%
The likelihood ratio test for $H_{0}:\theta _{1}=0$ and/or $\theta _{2}=0$
uses the maximized log-likelihood under the alternative (the same in both
cases) and under the null, which means minimizing over $LR_{\theta _{1}=0}$
and $LR_{\theta _{2}=0}$ and hence: 
\begin{equation*}
LR=\min \left\{ LR_{\theta _{1}=0},LR_{\theta _{2}=0}\right\} ,
\end{equation*}%
which is equivalent to rejecting for large values of: 
\begin{equation*}
\min \left\{ T_{1}^{2},T_{2}^{2}\right\} \text{ }\text{or}\min \left\{
\left\vert T_{1}\right\vert ,\left\vert T_{2}\right\vert \right\} .
\end{equation*}

\section{Invariance}

\label{sec:AppendixInvariance}

When testing the no-mediation hypothesis $H_{0}:\theta _{1}\theta _{2}=0,$
the parameters $\tau ,\sigma _{11},$ $\sigma _{22}$ are nuisance parameters
and their values have no influence on whether the null is true or not. %
\citet{HillierVanGarderenVanGiersbergen2021} shows that $T$ is maximal
invariant with respect to an appropriate group of transformations that
leaves the testing problem invariant and provides justification for
restricting attention to the two $t$-statistics. These exact invariance
results provide a strong justification for restricting attention to the two $%
t$-statistics for any sample size, finite or asymptotically, since it is
natural to restrict the problem to procedures that are scale invariant and
do not depend on $\tau .$

The testing problem has further invariance properties. The problem is
invariant to changing the signs (reflections) of $T_{1}$ and $T_{2}$ or
permuting them. This leads to maximal invariants with a sample and parameter
space that is only part of $%
\mathbb{R}
^{K}$.

\begin{proof}
(of Theorem \ref{Th:MaxInvAbsOrderStatK}) $\left\{ \left\vert T\right\vert
_{\left( 1\right) },...,\left\vert T\right\vert _{\left( K\right) }\right\} $
is obviously invariant to changes in sign and permutation as a consequence
of the absolute values and subsequent sorting. It is a maximal invariant
because any two $T$ and $\tilde{T}$ such that $\left\{ \left\vert
T\right\vert _{\left( 1\right) },...,\left\vert T\right\vert _{\left(
K\right) }\right\} =\left\{ |\tilde{T}|_{\left( 1\right) },...,|\tilde{T}%
|_{\left( K\right) }\right\} $ can only hold if $\tilde{T}$ is a permutation
of $T$ with a number of sign changes. Hence there will exist a
transformation $\mathbf{g}=\mathbf{g}_{1}\cdot \mathbf{g}_{2}\in G_{1}\times
G_{2}$ s.t. $\tilde{T}=\mathbf{g}\cdot T.$ The same argument holds for $%
\left\{ \left\vert \mu \right\vert _{\left( 1\right) },...,\left\vert \mu
\right\vert _{\left( K\right) }\right\} $ since the group of transformations
on the parameter space is the same as on the sample space. That the
distribution of $\left\{ \left\vert T\right\vert _{\left( 1\right)
},...,\left\vert T\right\vert _{\left( K\right) }\right\} $ depends only on $%
\left\{ \left\vert \mu \right\vert _{\left( 1\right) },...,\left\vert \mu
\right\vert _{\left( K\right) }\right\} $ is a property of a maximal
invariant.
\end{proof}

Lemma \ref{Lem:MaxInvAbsOrderStatKdistribution} gives an explicit expression
that further shows that the distribution is invariant under the $G_{1}\times
G_{2}$.

\begin{proof}
(of Lemma \ref{Lem:MaxInvAbsOrderStatKdistribution}) The absolute value of
the normal variate $T_{k}$ with mean $\mu _{k}$ and variance $1$ follows a
noncentral Chi-distribution with one degree of freedom. The $K$
distributions $\chi \left( \left\vert t\right\vert _{\left( k\right)
},\left\vert \mu \right\vert _{\left( k\right) }\right) $ are independent.
The result is then a direct application of \citet[eq. 6]{Vaughan1972}.
\end{proof}

\section{Proof of Theorem \protect\ref{Th:exactsimilartest}}

\label{sec:AppendixProofNonExistenceTheorem}

The proof exploits the symmetries of the problem and the completeness of the
normal distribution. We therefore consider the sample space of $T$ in $%
\mathbb{R}
^{2}$, rather than the octant that is the sample space of the maximal
invariant. The proof is constructive. It shows how to construct a monotonic
weakly increasing function $g\left( \cdot \right) $ with $\left\lfloor
1/\alpha \right\rfloor $ steps, when $1/\alpha $ is an integer. If $1/\alpha 
$ is not an integer then it cannot satisfy the condition that the final step
equals the asymptotic value $\lim_{t\rightarrow \infty }g\left( t\right) ,$
which is the normal critical value and determined by letting $\mu
\rightarrow \infty .$ There are two trivial solutions. First, if $g\left(
t\right) =0$ then $CR_{g}=%
\mathbb{R}
^{2}$ and the test always rejects, leading to an NRP of $100\%$ for all
parameter values. The other trivial solution is $g\left( t\right) =t$ such
that $AR_{g}=%
\mathbb{R}
^{2}$ and the test would never reject and the NRP is 0\% for all $\mu $.

\bigskip

\begin{proof}
Symmetry of the problem implies for the boundary function $g\left( -t\right)
=g\left( t\right) =g\left( \left\vert t\right\vert \right) $ and $g\left(
t\right) \leq t.$ So we only define $g\left( t\right) $ for $t\geq 0\,$\ and 
$g\left( 0\right) =0.$ The boundary function $g\left( t\right) $ is weakly
monotonically increasing by assumption, but may be a step function. If $%
g\left( t\right) $ is a step function then it has no ordinary inverse but we
can define its generalized inverse as:%
\begin{equation*}
g^{-1}\left( t\right) =\inf \left\{ x\mid g\left( x\right) >t\right\}
\end{equation*}%
(cf. quantile function e.g. \citet[p.304]{Vaart2000}). This definition of $%
g^{-1}\left( t\right) $ with strict inequality is chosen such that a
necessary condition for similarity may hold.

\begin{enumerate}
\item For any finite constant $cv_{1}$ we have 
\begin{equation*}
\lim_{\mu _{1}\rightarrow \infty }P\left[ \left\vert T_{1}\right\vert
>cv_{1}\mid \mu _{1}\right] =1
\end{equation*}%
and rejection only depends on $T_{2}$ when $\mu _{1}\rightarrow \infty $.
Hence $P\left[ \left\vert T_{2}\right\vert >c_{\infty }\right] =2P\left[
T_{2}>c_{\infty }\right] =2(1-\Phi \left( c_{\infty }\right) )=\alpha ,$
such that:%
\begin{equation*}
c_{\infty }=\Phi ^{-1}\left( 1-\frac{\alpha }{2}\right) .
\end{equation*}%
Monotonicity of $g$ implies $g\left( t\right) \leq \lim_{t\rightarrow \infty
}$ $g\left( t\right) =\Phi ^{-1}\left( 1-\frac{\alpha }{2}\right) $ and we
follow the convention that 
\begin{equation*}
g^{-1}\left( t\right) =+\infty \text{ if }t\geq \Phi ^{-1}\left( 1-\frac{%
\alpha }{2}\right) .
\end{equation*}

\item The null rejection probability as a function of $\mu $ equals:%
\begin{eqnarray*}
NRP_{\alpha }\left( \mu _{1}\right) &=&2\int_{-\infty }^{+\infty }\phi
\left( t_{1}-\mu _{1}\right) \left[ \int_{g\left( t_{1}\right)
}^{g^{-1}\left( t_{1}\right) }\phi \left( t_{2}\right) dt_{2}\right] dt_{1}
\\
&=&2\int_{-\infty }^{+\infty }\phi \left( t_{1}-\mu _{1}\right) \left[ \Phi
\left( g^{-1}\left( t_{1}\right) \right) -\Phi \left( g\left( t_{1}\right)
\right) \right] dt_{1}.
\end{eqnarray*}

\item Similarity requires $NRP_{\alpha }\left( \mu \right) =\alpha \forall
\mu \in 
\mathbb{R}
$. Hence%
\begin{eqnarray}
\alpha &=&2\int_{-\infty }^{+\infty }\phi \left( t_{1}-\mu _{1}\right) \left[
\Phi \left( g^{-1}\left( t_{1}\right) \right) -\Phi \left( g\left(
t_{1}\right) \right) \right] dt_{1}\Rightarrow  \notag \\
0 &=&\int_{-\infty }^{+\infty }\phi \left( t_{1}-\mu _{1}\right) \left[ \Phi
\left( g^{-1}\left( t_{1}\right) \right) -\Phi \left( g\left( t_{1}\right)
\right) -\frac{\alpha }{2}\right] dt_{1}  \notag \\
&=&\int_{-\infty }^{+\infty }\phi \left( t_{1}-\mu _{1}\right) F\left(
t\right) dt_{1},\text{\ \ with}  \notag \\
F\left( t\right) &=&\Phi \left( g^{-1}\left( t\right) \right) -\Phi \left(
g\left( t\right) \right) -\frac{\alpha }{2}.  \label{eq:Ftis0appendix}
\end{eqnarray}

\item Completeness of the normal distribution with mean $\mu _{1}$ implies
the condition%
\begin{equation}
F\left( t\right) =0.  \label{eq:Fis0condition}
\end{equation}%
We will show that this leads to a step function and the proof iteratively
determines the stepping points and step sizes, illustrated in the Figure \ref%
{Fig:ExactSimBoundalfa25}.

\item Starting at $t=0$ we have $g\left( 0\right) =0$ by definition, but for
a generalized inverse $g^{-1}\left( 0\right) $ there are two possibilities:

\begin{enumerate}
\item $g^{-1}\left( 0\right) =0=g\left( 0\right) .$ For $t=0,$ condition (%
\ref{eq:Fis0condition}) implies: $F\left( 0\right) =\Phi \left( g^{-1}\left(
0\right) \right) -\Phi \left( g\left( 0\right) \right) -\frac{\alpha }{2}=$ $%
\Phi \left( 0\right) -\Phi \left( 0\right) -\frac{\alpha }{2}=-\frac{\alpha 
}{2}=0$ only holds if $\alpha =0$. \ In this case $F\left( t\right) =\Phi
\left( g^{-1}\left( t\right) \right) -\Phi \left( g\left( t\right) \right)
=0 $ can only occur if $g^{-1}\left( t\right) =g\left( t\right) =t$ and the
test never rejects.

\item $g^{-1}\left( 0\right) =c_{1}>0$. In this case

\begin{enumerate}
\item $g^{-1}\left( t\right) =\inf \left\{ x\mid g\left( x\right) >t\right\} 
$ \ and\ $g^{-1}\left( 0\right) =\inf \left\{ x\mid g\left( x\right)
>0\right\} $ imply $g\left( t\right) =0,\forall 0\leq t<c_{1}$ and $g\left(
c_{1}\right) >0.$

\item $g\left( t\right) =0\ \forall \ 0\leq t<c_{1}$ and $F\left( 0\right) $
imply $g^{-1}\left( t\right) =c_{1}$\ $\forall \ 0\leq t<c_{1}.$

\item $g^{-1}\left( t\right) =c_{1}$\ $\forall \ 0\leq t<c_{1}$ and $%
g^{-1}\left( t\right) =\inf \left\{ x\mid g\left( x\right) >t\right\} $ in
turn imply $g\left( c_{1}\right) \geq c_{1},$ but $g\left( t\right) \leq t$
so

\item $g\left( c_{1}\right) =c_{1}.$

\item $c_{1}=\Phi ^{-1}\left( \frac{1}{2}+\frac{\alpha }{2}\right) $ because 
$F\left( 0\right) =\Phi \left( c_{1}\right) -\Phi \left( 0\right) -\frac{%
\alpha }{2}=0$ implies $\Phi \left( g^{-1}\left( t\right) \right) =\frac{1}{2%
}+\frac{\alpha }{2}$ and the result follows by the uniqueness of the inverse
of $\Phi $.
\end{enumerate}
\end{enumerate}

\item This argument can now be repeated.

\begin{enumerate}
\item In $\Phi \left( g^{-1}\left( c_{1}\right) \right) -\Phi \left( g\left(
c_{1}\right) \right) -\frac{\alpha }{2}=0$ implies $\Phi \left( g^{-1}\left(
c_{1}\right) \right) =\left( \frac{1}{2}+\frac{\alpha }{2}\right) +\frac{%
\alpha }{2}$ and $g^{-1}\left( c_{1}\right) =\Phi ^{-1}\left( \frac{1}{2}+%
\frac{2\alpha }{2}\right) =c_{2}.$ As in 5.b.i)\ $g\left( t\right)
=c_{1}\forall c_{1}\leq t<c_{2}\Rightarrow $ ii) $g^{-1}\left( t\right)
=c_{2}$\ $\forall \ c_{1}\leq t<c_{2}$ $\Rightarrow $ iii) $g\left(
c_{2}\right) =c_{2}$ and $g^{-1}\left( c2\right) =\Phi ^{-1}\left( \frac{1}{2%
}+\frac{3\alpha }{2}\right) =c_{3}.$

\item After $J$ repetitions of this argument we obtain a step function:%
\begin{eqnarray*}
g\left( \left\vert t\right\vert \right) &=&c_{j}:c_{j-1}\leq t<c_{j},\ \ \
j=1,..,J \\
c_{0} &=&0,c_{j}=\Phi ^{-1}\left( \frac{1}{2}+j\frac{\alpha }{2}\right)
\end{eqnarray*}
\end{enumerate}

\item An exact similar test only exists if $1/\alpha $ is an integer.

\begin{enumerate}
\item If $R$ is an integer such that $\alpha =1/R$ then $c_{R}=\Phi
^{-1}\left( \frac{1}{2}+\frac{R\ \alpha }{2}\right) =\Phi ^{-1}\left(
1\right) =+\infty $ and $c_{R-1}=\Phi ^{-1}\left( \frac{1}{2}+\frac{\left(
R-1\right) \ \alpha }{2}\right) =\Phi ^{-1}\left( 1-\frac{\alpha }{2}\right)
=$ $\lim_{t\rightarrow \infty }$ $g\left( t\right) $. \ The resulting
step-function $g\left( t\right) $ is constructed such that $F\left( t\right)
=0$ for all $t\in 
\mathbb{R}
$. The NRP equals $2\int_{-\infty }^{+\infty }\phi \left( t_{1}-\mu \right) %
\left[ \Phi \left( g^{-1}\left( t_{1}\right) \right) -\Phi \left( g\left(
t_{1}\right) \right) \right] dt_{1}=\alpha $ for all $\mu .$ \newline
So an exact similar test exists if $1/\alpha $ is an integer.

\item If $1/\alpha $ is not an integer then define $R=\left\lfloor 1/\alpha
\right\rfloor $ (entier or floor function) $R\in 
\mathbb{N}
$ and $0<r<1$ the remainder such that $\alpha =1/\left( R+r\right) .$ So $%
R\alpha <1$. After $R$ iterations of the argument in Step 5 gives $\Phi
^{-1}\left( \frac{1}{2}+\frac{R\ \alpha }{2}\right) =\Phi ^{-1}\left( \frac{1%
}{2}+\frac{\left( R+r-r\right) \alpha }{2}\right) =\Phi ^{-1}\left( 1-\frac{%
r\alpha }{2}\right) >\Phi ^{-1}\left( 1-\frac{\alpha }{2}\right) .$
Similarly $R-1$ iterations give: $\Phi ^{-1}\left( \frac{1}{2}+\frac{\left(
R-1\right) \ \alpha }{2}\right) =\Phi ^{-1}\left( 1-\frac{\alpha }{2}-\frac{%
r\alpha }{2}\right) <\Phi ^{-1}\left( 1-\frac{\alpha }{2}\right) $.
Therefore the step-function $g\left( \cdot \right) $ does not have $%
\lim_{t\rightarrow \infty }g\left( t\right) =$ $\Phi ^{-1}\left( 1-\frac{%
\alpha }{2}\right) $ which is the requirement for the NRP to equal $\alpha $
as $\mu _{1}\rightarrow \infty .$ Further note that the next step after $R$
iterations has $\left( 1+\frac{(1-r)\alpha }{2}\right) >1$ so we cannot
apply $\Phi ^{-1}$ to obtain the next boundary point.\newline
So no similar $g$-test exists if $1/\alpha $ is not an integer.
\end{enumerate}
\end{enumerate}
\end{proof}

\section{Algorithms}

\label{sec:AppendiixAlgorithms}

The construction of the optimal $g$-test is in two steps. The first step is
a basic implementation of the general varying-$g$ method. This generates a
near similar test that deviates less than $0.01\%$ points from $5\%.$ We use
this $\epsilon $ as a starting value for determining an upper bound to the
power envelope. The second step is using this upper bound to derive an
optimal $g$-test that minimizes the distance between the power surface and
the power envelope for tests in $\Gamma _{\alpha ,\epsilon }^{\mathbb{M}%
_{0}} $. Implementation of the varying-$g$ method.

\textbf{Basic $g$-function algorithm}

\begin{enumerate}
\item Define $g\left( \cdot \right) $ nonparametrically as a linear spline
defined by $J+2$ knots $\left\{ \left( t^{\left( j\right) },g^{\left(
j\right) }\right) \right\} _{j=0}^{J+1},$ i.e. by $J+2$ values $g^{\left(
j\right) }$ on a grid of points $t^{\left( j\right) }.$ The first and last
knots are fixed at $(0,0)$ and $(2.5,z_{0.025})$ respectively, so there are $%
J$ knots to be chosen. One of those knots is chosen $t=z_{0.025}$ such that
the LR\ boundary can be constructed as initializing function. For points $t$
not on the grid, $g\left( t\right) $ is obtained by linear interpolation,
and $g\left( t\right) =z_{0.025}\approx 1.96$ for $t>2.5.$

\item The criterion function $Q\left( g\right) $ is the accumulated NRP
deviation from $5\%$ as measured by a quadratic loss function over a grid of
points $\left\{ \mu _{0}^{\left( \iota \right) }\right\} _{\iota
=1}^{\Upsilon _{0}}$ with $\Upsilon _{0}>J$ and $\mu _{0}^{\left( 1\right)
}=0:$%
\begin{align*}
Q\left( g\right) & =\sum_{\iota =1}^{\Upsilon _{0}}\left( NRP_{g}\left( \mu
_{0}^{\left( \iota \right) }\right) -0.05\right) ^{2}, \\
s.t.~~~NRP_{g}\left( \mu _{0}^{\left( \iota \right) }\right) & \leq 0.05\ \
\ \ with \\
NRP_{g}\left( \mu \right) & =P\left[ T\in CR_{g}\left\vert \mu _{1}=0,\mu
_{2}=\mu \geq 0\right. \right]
\end{align*}

\item Minimize $Q\left( g\right) $ by varying $g\left( \cdot \right) $:

\begin{enumerate}
\item Initialize $g\left( \cdot \right) $ with knots $\left\{ \left(
0,0\right) ,\left( 1.96,1.96\right) ,\left( 2.5,1.96\right) \right\} $\
which is the $J=1$ function corresponding to the LR\ boundary. The first and
last knot are fixed and the middle one is varied when optimizing $Q\left(
g\right) $.

\item For given $g\left( \cdot \right) $ calculate the NRPs by numerical
integration for the grid of $\Upsilon _{0}$ noncentrality parameter points $%
\{(0,\mu ^{\left( \iota \right) })\}_{\iota =1}^{\Upsilon _{0}}$ under the
null, with $\Upsilon _{0}\geq J$ and calculate $Q\left( g\right) .$

\item Vary $g\left( \cdot \right) $ by changing $J$ knots and minimize the
criterion function $Q\left( g\right) $, subject to:

\begin{enumerate}
\item $0\leq g^{\left( j+1\right) }-g^{\left( j\right) }<\delta \ $:
monotonicity and limited increase

\item $g\left( t\right) \leq t\ $: logical restriction since maximal
invariant is absolute order statistic and $\left\vert T\right\vert _{\left(
1\right) }\leq \left\vert T\right\vert _{\left( 2\right) }$

\item $g^{\left( J+1\right) }=z_{0.025}\ $: dimensional coherence requires
reduction to one-dimensional solution (see Section \ref{sec:HigherDimensions}%
)
\end{enumerate}

\item Increase the number of knots $J$ and iterate until
convergence.\smallskip
\end{enumerate}
\end{enumerate}

\textbf{Comments}

\begin{enumerate}
\item We set $g\left( t\right) =z_{0.025}$ for $t>2.5,$ because for large
enough $\left\vert T\right\vert _{2}$ it is essentially known that $\theta
_{2}\neq 0$ and the rejection depends only on whether $\theta _{1}=0$ is
rejected. The corresponding $5\%$ critical value for $\left\vert
T\right\vert _{1}$ based on the normal distribution is the usual $%
z_{0.025}\approx 1.96$ as $\left\vert T\right\vert _{2}\rightarrow \infty .$

\item For $J$ small there are big gains in reducing the deviation from $5\%$
by varying the knots $\left\{ t^{\left( j\right) },g^{\left( j\right)
}\right\} _{j=1}^{J}$ and also by increasing $J,$ see Figure \ref%
{fig:NRPcompareJ369LRW}.

\item The number $\Upsilon _{0}$ of $\mu ^{\left( \iota \right) }$ points to
check similarity was chosen to be $\Upsilon _{0}=76>J$: 60 points equally
spaced between 0 and 6, and 16 points equally spaced between 6 and 20. This
imposes 152 side conditions. Step 3(c) imposes a further $3J$ restrictions
approximately for every choice of $J$, and about $100$ when $J=32.$

\item We have actually used as a side condition $0.05-\epsilon \leq
NRP_{g}\left( \mu _{0}^{\left( \iota \right) }\right) \leq 0.05$ and
iterated to find the smallest $\epsilon $ that yields a feasible solution.
\end{enumerate}

\smallskip

\textbf{Optimal $g$-function}

In order to find the optimal $g$, we minimize the sum of differences between 
$g$'s power\ surface and the power envelope on a grid of points, subject to
the size and $\epsilon $-similarity conditions. We impose monotonicity $%
g\left( t^{\left( i+1\right) }\right) \geq g\left( t^{\left( i\right)
}\right) $ and, since by definition of the absolute order statistic $%
\left\vert T\right\vert _{\left( 1\right) }\leq \left\vert T\right\vert
_{\left( 2\right) }$,we logically restrict $g$ to $0\leq g\left( t\right)
\leq t$.

\vspace{1em}

\textbf{Optimal $g$-function algorithm}

\begin{enumerate}
\item Define $g\left( \cdot\right) $ nonparametrically as a linear spline
defined above.

\item Define the criterion function $Q_{\epsilon }^{\ast }\left( g\right) $
as the accumulated power difference over the triangular grid of points $%
\mathbb{M}_{1}=\left\{ \left( \mu _{1}^{\left( \gamma ,\kappa \right) },\mu
_{2}^{\left( \gamma ,\kappa \right) }\right) \right\} _{1\leq \gamma <\kappa
\leq \Upsilon _{1}}:$%
\begin{equation*}
Q_{\epsilon }^{\ast }\left( g\right) =\sum_{\kappa =1}^{\Upsilon
_{1}}\sum_{\gamma \leq \kappa }\bar{\pi}\left( \mu _{1}^{\left( \gamma
,\kappa \right) },\mu _{2}^{\left( \gamma ,\kappa \right) }\right) -P\left[
CR_{g}\mid \left( \mu _{1}^{\left( \gamma ,\kappa \right) },\mu _{2}^{\left(
\gamma ,\kappa \right) }\right) \right]
\end{equation*}

\item Minimize $Q_{\epsilon }^{\ast }\left( g\right) $ by varying $g\left(
\cdot \right) $:

\begin{enumerate}
\item Initialize with $g\left( \cdot \right) $ equal to the LR boundary or
the previously determined basic $g$-function$.$

\item For given $g\left( \cdot \right) $ calculate $Q_{\epsilon }^{\ast
}\left( g\right) $ by numerical integration$.$

\item Vary $g\left( \cdot \right) $ by changing $J$ knots and minimize the
criterion function $Q_{\epsilon }^{\ast }\left( g\right) $, subject to:

\begin{enumerate}
\item $0.05-\epsilon \leq P\left[ CR_{g}\mid \left( 0,\mu _{0}^{\left( \iota
\right) }\right) \right] \leq 0.05,~~\forall \iota =1,\cdots ,\Upsilon _{0}$
: near similarity and size restrictions

\item $0=g\left( 0\right) \leq g\left( t^{\left( j\right) }\right) \leq
g\left( t^{\left( j+1\right) }\right) \leq t^{\left( j+1\right) }$ :
monotonicity,

\item $g^{\left( j+1\right) }-g^{\left( j\right) }<t^{\left( j+1\right)
}-t^{\left( j\right) }:$ limited increase and derivative,

\item $g\left( t\right) \leq t$ : logical restriction since argument is
absolute order statistic,

\item $g^{\left( J+1\right) }=z_{0.025}$ : dimensional coherence.
\end{enumerate}

\item Increase the number of knots $J$ and iterate until convergence.
\end{enumerate}
\end{enumerate}

The basic implementation algorithm solved the optimal $g$-boundary by
minimizing $\epsilon .$ Once $\epsilon $ is determined, the current
algorithm is akin to solving a dual problem that uses $\epsilon $ for the
inequality restrictions and maximizes power. It minimizes the total
difference from the power envelope.

\vspace{2em}

\textbf{Power Envelopes}

We calculate two power envelopes: one for near similar tests in $\Gamma
_{\alpha ,\epsilon }$ and a second for nonsimilar tests. The algorithm for
calculating the power envelope is related to \cite{Chiburis2009} and
implemented in Julia, see \cite{Bezanson17julia}, using Gurobi, an
optimization package that can handle many side restrictions; see \cite%
{Gurobi}. We maximize power subject to size and near similarity restrictions
on a grid of ${\Upsilon }_{0}$ parameter points under the null: $%
0.05-\epsilon \leq NRP\left( \mu ^{\left( \iota \right) }\right) \leq 0.05$
for $\iota ={1,...,\Upsilon }_{0}$. The upper bounds ensure correct size, at
least for the points considered. The lower bounds constitute the near
similarity restriction. The power envelope is obtained by repeating this
maximization on a grid of points $\left( \mu _{1},\mu _{2}\right) $ under
the alternative.

For the nonsimilar power envelope we can discard the lower bound
restrictions $0.05-\epsilon \leq NRP\left( \mu ^{(\iota )}\right) $. The
power can only increase (or remain the same) and the difference between the
two different power envelopes is the power loss one suffers from insisting
on similarity. This turns out to be less than $2\%$ points and it should be
stressed that this overstates the loss since no single test achieves the
power envelope.

Denote the parameter space for the ordered absolute noncentrality parameter%
\newline
$\Xi =\left\{ \left( \mu _{1},\mu _{2}\right) \in \mathbb{%
\mathbb{R}
}^{+}\mathbb{\times 
\mathbb{R}
}^{+}\mid 0\leq \mu _{1}\leq \mu _{2}\right\} .$

We will use a bounded (triangular) subset of this octant $\Xi $ defined as%
\newline
$\overline{\Xi }=\left\{ \left( \mu _{1},\mu _{2}\right) \in \mathbb{%
\mathbb{R}
}^{+}\mathbb{\times 
\mathbb{R}
}^{+}\mid 0\leq \mu _{1}\leq \mu _{2}\leq \mu _{\max }\right\} $ and
partitioned it into a null and alternative parameter set\newline
$\overline{\Xi }_{0}=\left\{ \left( \mu _{1},\mu _{2}\right) \in \mathbb{%
\mathbb{R}
}^{+}\mathbb{\times 
\mathbb{R}
}^{+}\mid 0=\mu _{1}\leq \mu _{2}\leq \mu _{\max }\right\} $ and\newline
$\overline{\Xi }_{1}=\left\{ \left( \mu _{1},\mu _{2}\right) \in \mathbb{%
\mathbb{R}
}^{+}\mathbb{\times 
\mathbb{R}
}^{+}\mid 0<\mu _{1}\leq \mu _{2}\leq \mu _{\max }\right\} $ respectively.

Analogously define the sample space of the maximal invariant/absolute order
statistic as $\mathbb{T=}\left\{ \left( t_{1},t_{2}\right) \in 
\mathbb{R}
_{0}^{+}\mathbb{\times }%
\mathbb{R}
_{0}^{+}\mid t_{1}\leq t_{2}\right\} $. Very large values of $t_{1}$ and $%
t_{2}$ are of limited interest and for computational purposes we can
restrict ourselves to a bounded triangular subset of the sample space: $%
\overline{\mathbb{T}}=\left\{ \left( t_{1},t_{2}\right) \in 
\mathbb{R}
_{0}^{+}\mathbb{\times }%
\mathbb{R}
_{0}^{+}\mid t_{1}\leq t_{2}\leq t_{\max }\right\}. $ \newline

\textbf{Power Envelope Algorithm}

\begin{enumerate}
\item Discretize $\overline{\Xi}_{0}$ into $\Upsilon_{0}$ points under $%
H_{0}:$ $\mathbb{M}_{0}=\left\{ \left( 0,\mu_{0}^{\left( \iota\right)
}\right) \right\} _{\iota=1}^{\Upsilon_{0}}.$

\item Discretize $\overline{\Xi }_{1}$ by choosing a triangular array of $%
\frac{1}{2}\Upsilon _{1}\left( 1+\Upsilon _{1}\right) $ points under $H_{1}:$
$\mathbb{M}_{1}=\left\{ \left( \mu _{1}^{\left( \gamma ,\kappa \right) },\mu
_{2}^{\left( \gamma ,\kappa \right) }\right) \right\} _{1\leq \gamma \leq
\kappa \leq \Upsilon _{1}}$

\item Partition $\overline{\mathbb{T}}$ into squares $s_{ij}$ with $1\leq
i\leq j\leq N$ such that $\dbigcup\limits_{1\leq i\leq j\leq N}s_{ij}=%
\overline{\mathbb{T}}$ and $s_{ij}\cap s_{kl}=\varnothing \ \ \forall
(i,j)\neq \left( k,l\right) .$

\item Under $H_{0}$ for $\iota=1,\cdots,\Upsilon_{0}$ calculate $%
p_{ij}^{\iota}=P\left[ \left( \left\vert T\right\vert _{\left( 1\right)
},\left\vert T\right\vert _{\left( 2\right) }\right) \in s_{ij}\mid\left(
0,\mu_{0}^{\left( \iota\right) }\right) \in\overline{\Xi}_{0}\right] $

\item For each $1\leq\gamma,\kappa\leq m$ choose $\mu^{\left( \gamma
,\kappa\right) }=\left( \mu_{1}^{\left( \gamma,\kappa\right) },\mu
_{2}^{\left( \gamma,\kappa\right) }\right) \in\mathbb{M}_{1}$ \ under the
alternative. For this $\mu$:

\begin{enumerate}
\item Calculate $p_{ij}^{\gamma\kappa}=P\left[ \left( \left\vert
T\right\vert _{\left( 1\right) },\left\vert T\right\vert _{\left( 2\right)
}\right) \in s_{ij}\mid\mu^{\left( \gamma,\kappa\right) }=\left( \mu
_{1}^{\left( \gamma,\kappa\right) },\mu_{2}^{\left( \gamma,\kappa\right)
}\right) \right] $ for each $s_{ij}\in\overline{\mathbb{T}}$

\item Determine the critical region to maximize the power 
\begin{equation*}
\underset{\left\{ \phi _{ij}^{\gamma \kappa },1\leq i\leq j\leq N\right\} }{%
\max }\sum_{1\leq i\leq j\leq N}p_{ij}^{\gamma \kappa }\phi _{ij}^{\gamma
\kappa }
\end{equation*}%
by selecting indicators $\phi _{ij}^{\gamma \kappa }=\mathbf{1}_{CR}\left(
s_{ij}\right) $ equal to 1 if $s_{ij}$ is part of the critical region, or 0
if part of the acceptance region, subject to the near similarity and size
restrictions on the NRPs:%
\begin{equation*}
0.05-\epsilon \leq \sum_{1\leq i\leq j\leq N}p_{ij}^{\iota }\phi _{ij}^{\mu
}\leq 0.05\text{ for }\iota ={1,...,\Upsilon }_{0}
\end{equation*}
\end{enumerate}
\end{enumerate}

\textbf{Comments}. Optimizer: Gurobi: $t_{\max }=11,$ each square $s_{ij}$
has lengths $0.01.$ Hence cardinality of $|\overline{\mathbb{T}}|=285150$.
For power calculations we use for $\mathbb{M}_{1}$ a regular grid with $\mu
_{2}\in \{0.2,0.4,\cdots ,4\},\mu _{1}\in \{0.2,0.4,\cdots ,\mu _{2}\}$. For
size and near similarity restrictions we use $\mu _{0}\in \{0.0,0.1,\cdots
,7.5\}$ and for near similarity $\epsilon =10^{-5}.$


\section{$g$-Function R Code}

\label{sec:AppendixRCode} \textbf{R Code}

\texttt{g <- function(t)\{}

\texttt{\ \ tabs=abs(t)}

\texttt{\ \ x \TEXTsymbol{<}- c(0., 0.1, 0.11, 0.13, 0.14, 0.15, 1.35, 1.36,
1.37, 1.44, 1.45, 2.05, 2.06, 2.07, 2.08, 2.09, 2.1)}

\texttt{\ \ y \TEXTsymbol{<}- c(0., 0.1, 0.106723, 0.106723, 0.106724,
0.106724, 1.30583, 1.31286, 1.3131, 1.3131, 1.3175, 1.9175, 1.9275, 1.9375,
1.9475, 1.9575, 1.95996)}

\texttt{\ \ ifelse(tabs\TEXTsymbol{>}=2.1,1.95996,approx(x,y,xout=tabs)\$y)}

\texttt{\ \} }

\end{appendices}%

\bibliographystyle{chicago}
\bibliography{REFERENCES}

\end{document}